\documentclass[a4paper,11pt,oneside]{article}
\usepackage[a4paper, margin=1in]{geometry}
\usepackage{setspace}
\usepackage{textcomp}
\usepackage{enumerate}
\usepackage{amsmath, amssymb,amsfonts,amsthm,mathrsfs}
\usepackage{stmaryrd} 
\usepackage{algpseudocode,algorithm}
\usepackage{cancel, cases}
\usepackage{soul}
\usepackage{float,graphicx,subcaption,epstopdf}
\usepackage{longtable,tabu,threeparttable,rotating,adjustbox,multicol,multirow}
\usepackage[normalem]{ulem}
\usepackage[dvipsnames]{xcolor}
\usepackage{tcolorbox} 

\usepackage{hyperref}

\usepackage{setspace,caption}

\usepackage[round, authoryear]{natbib}

\newcommand\tr{\text{tr}}

\newcommand{\bi}{\begin{itemize}}
\newcommand{\ei}{\end{itemize}}

\theoremstyle{plain}
\newtheorem{theorem}{Theorem}
\newtheorem{corollary}[theorem]{Corollary}
\newtheorem{lemma}[theorem]{Lemma}
\newtheorem{proposition}[theorem]{Proposition}

\theoremstyle{definition}
\newtheorem{definition}[theorem]{Definition}
\newtheorem{assumption}[theorem]{Assumption}

\theoremstyle{remark}
\newtheorem{remark}{Remark}

\numberwithin{equation}{section}
\numberwithin{theorem}{section}

\title{Exploratory Randomization for Discrete-Time Risk-Sensitive Benchmarked Investment Management with Reinforcement Learning}
\author{S\'ebastien Lleo\footnote{Finance Department and `AI, Data Science \& Business' AE, NEOMA Business School, France.} \and Wolfgang Runggaldier\footnote{University of Padova, Italy, and Fellow, Institut Louis Bachelier, Paris, France.}}
\date{\today}

\begin{document}
\maketitle

\begin{abstract}
This paper bridges reinforcement learning (RL) and risk-sensitive stochastic control by introducing a tractable exploration mechanism for policy search in risk-sensitive portfolio management, with known and unknown model parameters, that yields an endogenous relative-entropy regularization. We construct a discrete-time risk-sensitive benchmarked investment model. This model combines a factor-based asset universe with periodic portfolio rebalancing. Exploration is incorporated through user-specified Gaussian perturbations to baseline (exploitative) controls. The risk-sensitive stochastic control problem is solved analytically using the Free Energy-Entropy Duality. The Duality recasts the control problem as a linear-quadratic-Gaussian game and introduces a natural penalty for exploration. This approach yields simple sufficiency conditions for optimality. It also induces intuitive bounds on exploration based on risk sensitivity, asset covariance, and rebalancing frequency. Additionally, the optimal investment strategy can be interpreted through the lens of fractional Kelly strategies. By connecting risk-sensitive control theory and RL, this work provides a principled parametric family for policy-gradient implementations, guiding the design of RL methods.
\end{abstract}

\maketitle


\noindent\textbf{keywords}: risk-sensitive control, fractional Kelly strategies, free energy-entropy duality, stochastic games, Linear Quadratic Gaussian games, exploratory controls, portfolio optimization.


\noindent\textbf{MSC classes}: 68T05; 91G10; 91G80; 93E20; 93E35.

\section{Introduction}

This paper addresses the gap between reinforcement learning and risk-sensitive stochastic control by introducing a tractable exploration mechanism for policy search in risk-sensitive portfolio problems, with known and unknown model parameters, that yields an endogenous relative-entropy regularization. The mechanism mirrors the exploration–exploitation trade-off in reinforcement learning. Reinforcement learning and stochastic control, including risk-sensitive control, trace their roots to Bellman's dynamic programming principle, but their divergent views on the availability of information and the necessity of exploration have created a significant gap between the two approaches. Reinforcement learning considers environments where there is insufficient information at the outset to estimate a model. The exploration-exploitation trade-off then becomes an essential tool for balancing the pursuit of new information (exploration) with the achievement of a goal based on existing knowledge (exploitation). On the other hand, the theory of stochastic control generally focuses on settings with perfect information about the model's coefficients. In such settings, exploration is typically unnecessary; the optimal policy is purely exploitative. Our work aims to bridge these perspectives within the setting of risk-sensitive portfolio management.

We start from a factor model for an investment universe with $m$ continuously traded assets. Drifts depend on $n$ valuation factors estimated at discrete times. Investors rebalance portfolios when factors are estimated, leading to a discrete-time control problem with continuous-time asset dynamics between rebalancing dates. This hybrid continuous/discrete-time approach, similar to that of \citet{stettnerRiskSensitivePortfolio1999}, is consistent with the empirical asset pricing literature, which usually estimates factors monthly, even though assets trade throughout the day. Then, we introduce an investment benchmark, as in \citet{dall_RSBench,LleoRunggaldierSeparation24}, serving as a performance measure that investors aim to outperform, depending on their degree of risk sensitivity. Finally, we incorporate exploration as a randomized control. Exploration is modeled as a user-specified Gaussian perturbation to the baseline (exploitative) control. This modelling choice is consistent with the treatment in \citet{LleoRunggaldier_ExploratoryRandomization_2026}, and with recent literature on continuous-time reinforcement learning \citep{wangReinforcementLearningContinuous2020,wangContinuousTimeMean2020,jiaQLearningContinuousTime,jiaPolicyEvaluationTemporalDifference2022,jiaPolicyGradientActor2022,jiaContinuoustimeRisksensitiveReinforcement2024a}.  

Then, we formulate the benchmark investment management model as a risk-sensitive control problem. Risk-sensitive control possesses appealing features for financial applications: it connects neatly to utility theory, provides a dynamic version of mean–variance optimization, converges to the Kelly criterion in the limit as the risk sensitivity $\theta \to 0$, and, for $\theta>0$, the risk-sensitive criterion coincides with the entropic risk measure \citep{lleoDualDominanceHow2025}. However, risk-sensitive investment management problems are not standard, so they cannot be solved directly. Our paper proposes a new solution technique that draws on the Free Energy–Entropy Duality (FEED) of \citet{daipraConnectionsStochasticControl1996}. We use the FEED to recast the non-standard control problem as an equivalent linear-quadratic-Gaussian (LQG) stochastic game under an equivalent measure, with a relative-entropy-based penalty. This penalty measures the cost implied by the distance between the physical probability measure and the dual measure, as well as the cost of exploration. Under verifiable sufficient conditions, we then derive an explicit analytical solution of the game: we establish the existence of a saddle point, obtain the optimal controls, and provide backward recursions for the value-function coefficients.

Our solution technique differs from \citet{stettnerRiskSensitivePortfolio1999}, who used a duality argument only to a limited extent to rewrite the original risk-sensitive control problem into one amenable to solution via contraction maps, and from Kuroda and Nagai (2002), who applied a change of measure to transform the risk-sensitive control problem into a linear-exponential-of-quadratic (LEQG) problem. It is similar to \citet{LleoRunggaldier_ExploratoryRandomization_2026}, who also apply the FEED to solve a generic, but standard, LEQG problem. By comparison, our approach fully exploits the FEED to produce an analytical solution to our non-standard risk-sensitive control problem with exploratory controls. Our approach is also advantageous because the penalty for exploration arises naturally as a byproduct of the solution technique, whereas in the continuous-time reinforcement learning literature, this penalty must be specified in the model.

Our approach makes four main contributions: it uncovers intuitive insights into the role of exploration, formulates easy-to-check sufficiency conditions for the existence of a solution, interprets the asset allocation, and reveals a clear connection between risk-sensitive control and reinforcement learning. 

First, our approach shows that optimal exploration is unbiased and provides an explicit upper bound on the user-specified exploration covariance. The FEED naturally treats the mean of the exploration distribution as a control. We show that at the saddle point, this control equals 0, so optimal exploration is unbiased. Regarding the user-specified covariance, the sufficient conditions provide an intuitive upper bound that depends inversely on the investor's risk sensitivity, the diffusion matrix of the risky assets, $\Sigma\Sigma'$, and the time interval between two rebalancing dates. Investors with greater risk sensitivity explore less, while those with lower risk sensitivity explore more, as one would expect. In the limiting case, Kelly investors face no risk-sensitivity-induced bounds on exploration. Moreover, investors should be more cautious with exploration when the covariance of financial asset returns is high. Intuitively, in a factor model, investors buy assets to achieve optimal exposure to the underlying factors. Thus, asset return volatility creates variability around the factor exposure, resulting in natural exploration. The exploration covariance bound accounts for this exploration naturally. Additionally, investors with shorter time between rebalancing dates can explore more. When rebalancing occurs more frequently, investors learn more from exploration, and they can correct suboptimal asset allocations more rapidly.

Second, the main characterization of the optimal controls provides simple and easily interpretable sufficient conditions: the investment universe cannot have redundant assets (a standard condition in portfolio optimization), and the projected curvature of the value function restricted to the noise subspace has all eigenvalues strictly smaller than one (a standard condition for LQG control problems).  Using an alternative characterization of the optimal controls, we also show that these simple sufficient conditions can be recast as a different set of sufficient conditions articulated around the risk-resistance condition \citep[a standard condition in the risk-sensitive control literature][]{ShaijuPetersenFormulasLQR_LQ_LEQG2008,Whittle1990}.  

Third, using our alternative strategy characterization, we interpret optimal portfolio allocations through the lens of fractional Kelly strategies.  In continuous-time investment, fractional Kelly strategies are a fixed-mix investment in the Kelly portfolio, a benchmark-tracking portfolio, and an intertemporal-hedging portfolio \citep{dall_RSBench,LleoRunggaldier_EntropyRegularizationinRLandRSIM_2026}. In our discrete-time setting, the optimal investment is a rotated and rescaled fractional Kelly strategy, with the mix determined by time-dependent matrices. Moreover, these matrices reflect the risk-resistance condition and the LQG curvature condition. 

Finally, our approach bridges risk-sensitive control and reinforcement learning. We illustrate this idea in the policy gradient framework of \citet{hamblyPolicyGradientMethods2021,hamblyRecentAdvancesReinforcement2023} and in an actor-critic approach. If our market and model parameters could be estimated with high accuracy, we could solve the risk-sensitive control problem directly using the recursions in Theorem \ref{theo:main:recursions}. However, in practice, estimating the model parameters is notoriously difficult, especially the expected returns. Reinforcement learning bypasses the traditional statistical estimation of the model parameters, learning a parameterized form of the policy and value functions directly from market data. Therefore, the main purpose of our study of the risk-sensitive benchmarked investment model is to capture the essential features of financial markets, to characterize their fundamental relations, and to provide a principled parametric family for policy-gradient implementations. These insights are invaluable in guiding the design of reinforcement learning methods.

 \section{Setting Up the Risk-Sensitive Benchmarked Asset Management Problem}\label{sec:setup}

Let $\left(\Omega,  \mathcal{F}, \left(\mathcal{F}_s\right)_{s \in [0,T]}, \mathbb{P} \right)$ be a filtered complete probability space. We introduce
a $\mathbb{R}^d$-valued $\mathcal{F}_s$-Wiener process $W_s$ with components $W^i_s, i = 1, \ldots d$, with $d := n + m + 1$, where $n\geq 1$ is the number of factors, $m\geq1$ is the number of financial assets, and the additional dimension accounts for the residual stochastic component of the benchmark return that cannot be replicated through the traded assets\footnote{If the benchmark is perfectly spanned by the traded assets, no additional source of uncertainty is needed and one may take $d:=n+m$.}.

We consider an investor who rebalances her portfolio on a fixed regular schedule (e.g., daily, weekly, or monthly).  Let $0 = t_0 < t_1 < \dots < t_k < \ldots < t_K = T$ denote the rebalancing dates, and let $\Delta t := t_{k+1}-t_k$. We model the investor's self-financing investment \emph{strategy} as a piecewise constant, right continuous  $\mathcal{F}_s$-adapted process $H = \left(h_s\right)_{s \in [0,T]} \in \mathbb{R}^m$ such that the \emph{control} $h_s = h_k$ on each interval $[t_k,t_{k+1}), k=0,\ldots,K-1$ and $h_k$ is chosen at time $t_k$.

\subsection{Model for the Financial Market}\label{sec:RSBAM:model}

We adopt a hybrid-time specification that combines discrete-time factor dynamics, continuous-time asset prices, and discrete-time control updates. Factors are observed, and controls are updated at the rebalancing dates, while asset and benchmark prices evolve continuously between rebalancing times.
\\

\noindent\textbf{1. Risky Financial Assets}\\
The investor can trade $m$ risky financial assets $\left(S^1_s, \ldots, S^m_s\right) =: \left(S_s\right)_{s \in [0,T]}$ to constitute a portfolio. In this paper, we assume that all prices, benchmark level, and wealth are expressed in discounted units, with discounting performed at the prevailing short-term interest rate. The prices of the risky financial assets follow geometric dynamics: 
\begin{align}\label{eq:dS}
    \frac{dS^i_t}{S^i_t} 
    = (a + A X_k)_i dt
    + \sum_{j=1}^{d} \sigma_{ij} dW^j_t,
    \quad
    i = 1, \ldots, m; \quad t \in [t_k,t_{k+1}), k=0,\ldots,K-1
\end{align}
where $a \in \mathbb{R}^m$, $A \in \mathbb{R}^{m\times n}$, and $\Sigma := (\sigma_{ij}) \in \mathbb{R}^{m \times d}$, and where the factor process $X_k$ is defined next. In particular, an application of It\^o's lemma shows that for $t \in [t_k,t_{k+1}), k=0,\ldots,K-1$, 
\begin{align}\label{eq:S}
    S^i_t 
    = S^i_{t_k} \exp \left\{ \left[(a + A X_k)_i  - \frac{1}{2} \sum_{j=1}^d \sigma_{ij}^2\right] (t-t_k)
    + \sum_{j=1}^{d} \sigma_{ij} (W^j_t - W^j_{t_k}) \right\},
    \quad
    i = 1, \ldots, m.
\end{align}

We assume that no two assets have an identical risk profile:   
\begin{assumption}\label{as:sigma:posdef}
    The matrix $\Sigma\Sigma'$ is positive definite.
\end{assumption}

\noindent\textbf{2. Factor Process}\\
The  $n$ valuation factors $\left(X^1_t, \ldots, X^n_t\right) =: \left(X_k\right)_{t \in \{t_0,\ldots,t_K\}}$ are typically empirical asset pricing factors, macroeconomic factors, fundamental factors, or statistical factors. Following \citet{stettnerRiskSensitivePortfolio1999}, we model the first difference of the factors as a discrete-time first-order autoregressive (AR(1)) process:
\begin{align}\label{eq:state} 
    X_{k+1} - X_k = (b + B X_k)\Delta t + \Lambda w_k,
\end{align}
where we used the notational shortcut $X_k$ for the value $X_{t_k}$ of the factor process at time $t_k$, and where $b \in \mathbb{R}^n$, $B \in \mathbb{R}^{n \times n}$, and $\Lambda \in \mathbb{R}^{n \times d}$. The noise term $w_{k} := (W_{k+1}-W_k)$ is the increment of the Brownian motion $W$, where $W_k$ denotes $W_{t_k}$, the value of the Brownian motion at time $t_k$. Consequently, $w_{k}$ is Gaussian with mean equal to 0 and covariance equal to $\Delta t I_d$, where $I_d$ is the $d$-dimensional identity matrix. Hence, the drift term is scaled by $\Delta t$, while the innovation term inherits covariance $\Delta t\,I_d$ from Brownian increments.

This assumption is less restrictive than it appears. The empirical asset pricing literature typically estimates its models with a monthly frequency. A monthly frequency is also a popular choice in the portfolio optimization literature. However, asset prices evolve continuously.

We clean up the notation by introducing $\tilde{B} := I_n + B \Delta t$ to express\footnote{No stability restriction on $\tilde B$ is needed for the finite-horizon analysis below. Stronger assumptions may be imposed in infinite-horizon extensions.} \eqref{eq:state} as
 \begin{align}\label{eq:state:AR} 
    X_{k+1} = b \Delta t + \tilde{B} X_k + \Lambda w_k.
\end{align}
We shall refer to the process $(X_{t_k})_{k =0,\ldots,K}$ as the \emph{state} process.
\\

\noindent\textbf{3. Benchmark Index}\\
The investor manages a portfolio of financial assets against a benchmark index, typically a financial index or a custom-built passive portfolio. We model the benchmark's (discounted) level $\left(L_s\right)_{s \in [0,T]}$ as:
\begin{align}\label{eq:dL}
\frac{dL_t}{L_t} = (c + C X_k)dt + \Xi' dW_t,
\quad t \in [t_k,t_{k+1}), k=0,\ldots,K-1
\end{align}
where $c \in \mathbb{R}$, $C \in \mathbb{R}^n$, and $\Xi\in \mathbb{R}^{d}$.

\subsection{Risk-Sensitive Control Problem}\label{sec:RSBAM:RSC_problem}

\noindent\textbf{1. Wealth Process}\\
The wealth process $\left(V_s\right)_{s \in [0,T]}$ solves the SDE 
\begin{align}\label{eq:Phi:byinterval}
    \frac{dV_s}{V_s} 
    &= h_k'\left(a+ A X_k\right)ds
    + h_k' \Sigma	 dW_s,
\end{align}
for $s \in [t_k,t_{k+1}), k=0,\ldots,K-1$ and where $h_k$ denotes the value of the control at time $t_k$. A precise admissibility definition for the control is given in Section \ref{sec:sol:exploratorycontrols}. Here $h_k$ denotes the vector of portfolio weights in discounted risky assets at time $t_k$. The residual wealth is invested in the num\'eraire asset, with a discounted price constant at 1. Unless otherwise stated, short-selling and leverage are allowed. Throughout this section, we assume controls are non-anticipative and satisfy the required integrability conditions.
\\

\noindent\textbf{2. Log Price Relative and Log Excess Return}\\
To measure the investment portfolio's cumulative outperformance relative to its benchmark, we first compute the \emph{log price relative} process $\left(R_s\right)_{s \in [0,T]}$, defined as $R_s := \ln \frac{V_s}{L_s}$. This process solves the SDE:
\begin{align}\label{eq:excess_return:byinterval}
dR_s &= \left[ 
  \left(- \frac{1}{2} h_k'\Sigma\Sigma'h_k 
  + h_k'a 
  + \frac{1}{2}\Xi'\Xi- c \right)  
  + \left(h_k' A - C \right)X_k\right] ds
  + \left(h_k'\Sigma - \Xi' \right) dW_s. 
\end{align}
for $s \in [t_k,t_{k+1}), k=0,\ldots,K-1$. Hence, at the terminal time $T$,
\begin{align}\label{eq:excess_return}
R_T =& R_0 + \sum_{k=0}^{K-1}   \int_{t_k}^{t_{k+1}}\left\{ 
    \left(- \frac{1}{2} h_k'\Sigma\Sigma'h_k 
    + h_k'a 
    + \frac{1}{2}\Xi'\Xi- c \right)  
    + \left(h_k' A - C \right)X_k\right\} ds
                            \nonumber\\
    &+ \sum_{k=0}^{K-1} \int_{t_k}^{t_{k+1}}\ \left(h_k'\Sigma - \Xi' \right) dW_s
                            \nonumber\\
    =& R_0 + \sum_{k=0}^{K-1} \left\{ 
    \left(- \frac{1}{2} h_k'\Sigma\Sigma'h_k 
    + h_k'a 
    + \frac{1}{2}\Xi'\Xi- c \right)  
    + \left(h_k' A - C \right)X_k\right\} \Delta t
                                \nonumber\\
    &+ \sum_{k=0}^{K-1} \left(h_k'\Sigma - \Xi' \right) w_k. 
\end{align}
Using a benchmarked log-relative process is natural for active portfolio management, as it directly evaluates cumulative outperformance net of the benchmark and yields a state reduction through elimination of wealth. From the log price relative process, we obtain the cumulative \emph{log excess return} of the portfolio over its benchmark as the difference $ R_T - R_0$ of log price relatives. Conditional on the sequence $\left\{X_k,h_k\right\}_{k=0}^{K-1}$, the increments of $R_t-R_0$ over the rebalancing intervals are independent Gaussian random variables. It follows that $R_T-R_0$ is conditionally Gaussian, so its conditional exponential moment admits a closed-form expression.
\\

\noindent\textbf{3. Risk-Sensitive Benchmarked
Criteria}\\
The investor seeks to maximize the risk-sensitive benchmarked criterion: 
\begin{align}\label{eq:J} 
J(H,\theta)
&:= -\frac{1}{\theta} \ln \mathbf{E} \left[ e^{-\theta (R_T-R_0)} \right]
                        \nonumber\\
&= -\frac{1}{\theta} \ln \mathbf{E} \left[ \exp\left\{ -\theta \sum_{k=0}^{K-1} \left\{ 
    \left(- \frac{1}{2} h_k'\Sigma\Sigma'h_k 
    + h_k'a 
    + \frac{1}{2}\Xi'\Xi- c \right)  
    + \left(h_k' A - C \right)X_k\right\} \Delta t
    \right.\right.
                                \nonumber\\
    &\left.\left.
    - \theta \sum_{k=0}^{K-1} \left(h_k'\Sigma - \Xi' \right) w_k\right\}\right]
,
\end{align}
where $\theta$ is the risk sensitivity parameter, $T < \infty$ is a fixed time horizon. Here, we focus on the risk-averse case $\theta>0$. Note that the state (factor) process $X_k$ is the sufficient state variable. Wealth has been eliminated through the log-relative formulation.

We also introduce the exponentially-transformed criterion $ I(H,\theta)$ to be minimized,
\begin{align}\label{eq:I}
I(H,\theta)
&:= e^{-\theta J(H,\theta)}
= \mathbf{E} \left[ e^{-\theta (R_T-R_0)}\right].
\end{align}

Together, the discrete-time state process at \eqref{eq:state}, and the criteria  \eqref{eq:J}-\eqref{eq:I} constitute a non-standard discrete-time risk-sensitive control problem. 
The criterion is non-standard in two respects: 1) the control does not steer the state dynamics; and 2) the control and Brownian increments interact inside the exponential payoff. Consequently, unlike in more standard formulations, the stochastic integral term cannot be eliminated by taking conditional expectations.
\\

\noindent\textbf{4. Randomized Exploratory Control}\\
Up to this point, $h_k$ denotes a standard $\mathcal{F}_{t_k}$-adapted control. We now introduce randomized controls to model exploration in the spirit of reinforcement learning. Specifically, we consider randomized controls of the form $h_k=\bar h_k+v_k$, where a baseline control $\bar h_k$ is perturbed by a Gaussian shock $v_k \in \mathbb{R}^{m}$, independent of $\mathcal{F}_{t_k}$, thereby capturing an exploitation, via $\bar h_k$, and exploration, via $v_k$, within the risk-sensitive benchmarked portfolio problem.

\begin{definition}[Randomized Exploratory Control]\label{def:randomizedcontrol}
     A piecewise-constant $\mathbb{R}^m$-valued control process $(h_s)_{s\in[0,T]}$ is a \emph{randomized exploratory control} if, for each $k=0,\dots,K-1$ and $s \in [t_k, t_{k+1})$, it admits the representation
    \begin{align}
        h_s =\bar h_k + v_k,
    \end{align}
    where $\bar h_k$ is $\mathcal{F}_{t_k}$-measurable, and $v_k \sim \mathcal{N}(0, \Psi_k)$, with deterministic positive definite covariance $\Psi_k \in \mathbb{R}^{m \times m}$, is serially independent across $k$ and independent of $\mathcal{F}_{t_k}$.   
\end{definition}

In particular, the perturbation $v_k$ is independent of $w_k$, and more generally independent of the Brownian motion $W$. Hence, exploration does not implicitly anticipate market noise. Additionally, admissible strategies remain non-anticipative and compatible with Markov dynamic programming. Hence, the randomized control is of the form $h_k = \bar h_k + v_k \sim \mathcal{N} \left(\bar h_k, \Psi_k\right)$  and it admits a representation as a measure $\pi\left(dh;\bar h_k\right) \sim \mathcal{N} \left(\bar h_k, \Psi_k\right)$.

\begin{remark}
Allowing the covariance matrix $\Psi_k$ to vary with time $t_k$ opens the possibility to reduce the breadth of exploration over time, consistent with the reinforcement learning literature. Exploration is more valuable in early times when little information is available. As more information becomes available, exploration gradually loses its value.
\end{remark}

\subsection{Free Energy, Relative Entropy, and the Energy-Entropy Duality}\label{sec:energy-entropy:defs}

The following notions from \citet{daipraConnectionsStochasticControl1996} will be used in Section \ref{sec:sol:FEED} to reformulate the risk-sensitive control problem with criterion $I(H,\theta)$ as an equivalent Linear-Quadratic-Gaussian stochastic differential game. 

\begin{definition}[Free Energy]
The \emph{free energy} of a random variable $\psi$ with respect to a reference measure $\mathbb{P}$ is defined, under the integrability condition $\mathbf{E}^{\mathbb{P}}[e^\psi] < \infty$, as
\begin{align}\label{eq:free_energy}
  \mathcal{E}^{\mathbb{P}}\{\psi\} := \ln \mathbf{E}^{\mathbb{P}} \left[e^\psi\right] 
  = \ln \left( \int e^\psi \, d\mathbb{P} \right).
\end{align}
\end{definition}

\begin{definition}[Relative Entropy / Kullback-Leibler Divergence]
Consider a probability measure $\mathbb{P}^\gamma$ that is \emph{absolutely continuous} with respect to $\mathbb{P}$. The \emph{relative entropy} of $\mathbb{P}^\gamma$ with respect to $\mathbb{P}$ is
\begin{align}\label{eq:relative_entropy}
  D_{\rm KL}\left(\mathbb{P}^\gamma \,\|\, \mathbb{P}\right)
  := \mathbf{E}^{\gamma} \left[ \ln \frac{d\mathbb{P}^\gamma}{d\mathbb{P}} \right]
  = \mathbf{E}^{\gamma} \left[ \ln \mathcal{L}^\gamma \right],
\end{align}
where $\mathcal{L}^\gamma := \frac{d\mathbb{P}^\gamma}{d\mathbb{P}}$ denotes the Radon-Nikodym derivative and $\mathbf{E}^{\gamma}[\cdot]$ is the expectation under $\mathbb{P}^\gamma$.
\end{definition}

\begin{proposition}[Free Energy-Entropy Duality, {\cite[Prop.~2.3(ii)]{daipraConnectionsStochasticControl1996}}]\label{prop:FEEDuality}
Let $\psi$ be a measurable random variable such that $\mathbf{E}^{\mathbb{P}}[e^\psi] < \infty$. Then
\begin{align}\label{eq:free-energy-entropy-duality}
  \mathcal{E}^{\mathbb{P}}\{\psi\} 
  = \sup_{\mathbb{P}^\gamma \ll \mathbb{P}} 
    \left\{ 
      \mathbf{E}^{\gamma}[\psi] - D_{\rm KL}(\mathbb{P}^\gamma \,\|\, \mathbb{P})
    \right\},
\end{align}
where the supremum is taken over all probability measures $\mathbb{P}^\gamma$ absolutely continuous with respect to $\mathbb{P}$.
\end{proposition}

\begin{remark}
Under suitable conditions, the supremum in \eqref{eq:free-energy-entropy-duality} is attained at a measure $\mathbb{P}^{\gamma^*}$ whose Radon-Nikodym derivative with respect to $\mathbb{P}$ is given by
\begin{align}\label{eq:RNderivative-gamma-optimal}
  \mathcal{L}^{\gamma^*} := \frac{d\mathbb{P}^{\gamma^*}}{d\mathbb{P}} 
  = \frac{e^\psi}{\mathbf{E}^{\mathbb{P}}[e^\psi]}.
\end{align}
This expresses the optimal change of measure in terms of the exponential of the random variable $\psi$ and ensures that the transformed measure is properly normalized.
\end{remark}

\section{Solution Via The Free Energy-Entropy Duality}\label{sec:solution}

We now derive a solution for the randomized risk-sensitive control problem associated with the criteria \eqref{eq:J}-\eqref{eq:I} and the state dynamics at \eqref{eq:state:AR}. As announced above, we use the energy-entropy duality to associate with the risk-sensitive control problem, set with respect to an initial measure $\mathbb{P}$, a risk-neutral randomized stochastic game, set with respect to a transformed measure $\mathbb{P}^{\bar\gamma, \bar\eta}$. This risk-neutral randomized stochastic game turns out to be penalized by an appropriate relative entropy term.

As a reference for the reader, Table~\ref{tab:variables:summary} lists the key variables and parameters used in Sections 2 and 3.

\begin{table}
    \centering
    \resizebox{\textwidth}{!}{\begin{tabular}{ccclc}
        \hline\\
        Variable/ & Dimension   &  Update time   & Definition & Introduced \\
        parameter & & (discrete/continuous) & & in section \\
        \hline\\

        \multicolumn{5}{l}{\underline{Market variables:}} \\
         
        $X_k$  & $\mathbb{R}^{n}$   &    Discrete      & Valuation factor (state variable) & 2.1\\
        $S_s$  & $\mathbb{R}^{m}$   &    Continuous    & Financial asset discounted price & 2.1\\
        $L_s$  & $\mathbb{R}$       &    Continuous    & Benchmark level & 2.1\\
        $V_s$  & $\mathbb{R}$       &    Continuous    & Wealth process & 2.1\\
        $R_s$  & $\mathbb{R}$       &    Continuous    & Log price relative $R_s := \ln \left( V_s / L_s\right) $ & 2.1\\
        
        \hline\\
         
        \multicolumn{2}{l}{\underline{Under the measure $\mathbb{P}$:}} & 
        \multicolumn{3}{l}{$X_{k+1} = b\Delta t + \tilde{B} X_k + \Lambda w_k$}\\

         \multicolumn{2}{l}{ } &
         \multicolumn{3}{l}{\small $R_T - R_0 = \sum_{k=0}^{K-1} \left\{ 
            \left(- \frac{1}{2} h_k'\Sigma\Sigma'h_k 
            + h_k'a 
            + \frac{1}{2}\Xi'\Xi- c \right)  
            + \left(h_k' A - C \right)X_k\right\} \Delta t
            + \sum_{k=0}^{K-1} \left(h_k'\Sigma - \Xi' \right) w_k$} \\

        \\
        
        $W_s$ & $\mathbb{R}^{d}$ & Continuous  & $\mathcal{F}_s$-Wiener process on $\left(\Omega,  \mathcal{F}, \left(\mathcal{F}_s\right)_{s \in [0,T]}, \mathbb{P} \right)$ & 2 \\
        & $d=n+m+1$ & & & \\
   
        $h_k$ & $\mathbb{R}^{m}$ & Discrete  & 1. Generic control for the initial problem & 2.1 \\
                  &  & & 2. Exploratory control: $h_k = \bar{h}_k + v_t \sim \mathcal{N} \left(\bar{h}_k, \Psi_k\right)$  & 2.2 \\
        $\bar{h}_t$ & $\mathbb{R}^{m}$ & Discrete  & Deterministic control & 2.2 \\
        $v_k$   & $\mathbb{R}^{m}$ & Discrete  & Exploration noise: & 2.2 \\
                &  & & $v_k \sim \mathcal{N}\left(0,\Psi_k \right)$ & \\
        
        $J(H,\theta)$ & $\mathbb{R}$ & Discrete & Control criterion: $J(H,\theta):= -\frac{1}{\theta} \ln \mathbf{E} \left[ e^{-\theta (R_T-R_0)} \right]$ & 2.1 \\   
        $I(H,\theta)$ & $\mathbb{R}$ & Discrete & Control criterion: $I(H,\theta) := \exp \left\{-\theta J(H,\theta)\right\} = \mathbf{E} \left[ e^{-\theta (R_T-R_0)}\right]$ & 2.1 \\
        
        \\

                  \hline\\
         \multicolumn{2}{l}{\underline{Under the measure $\mathbb{P}^{\bar\gamma,\bar\eta}$:}} &
         \multicolumn{3}{l}{$X_{k+1} = b \Delta t + \tilde{B} X_k +\Lambda \bar\gamma_k \Delta t + \Lambda w^{\bar\gamma}_k$}\\

         \multicolumn{2}{l}{ } &
         \multicolumn{3}{l}{\small $\bar R^\pi_T - R_0 = \sum_{k=0}^{K-1}  \Bigg\{
            \left(- \frac{1}{2} \left(\bar h_k +\bar\eta_k\right)'\Sigma\Sigma'\left(\bar h_k +\bar\eta_k\right)
             - \frac{1}{2} \tr\left( \Psi_k\Sigma\Sigma'\right) 
            + \left(\bar h_k +\bar\eta_k\right)'a 
            + \frac{1}{2}\Xi'\Xi- c \right)$} \\
        \multicolumn{2}{l}{ } &
        \multicolumn{3}{l}{\small $ \phantom{R_T - R_0 =} 
            + \left(\left(\bar h_k +\bar\eta_k\right)' A - C \right)X_k 
            + \left((\bar h_k + \bar\eta_k)'\Sigma - \Xi' \right)\bar\gamma_k
            \Bigg\} \Delta t
            + \sum_{k=0}^{K-1} \left((\bar h_k + \bar\eta_k)'\Sigma - \Xi' \right) w^{\bar\gamma}_k.$} \\

        \\ 
         
        $\bar{\gamma}_k$ & $\mathbb{R}^{d}$  &  Discrete   & Change of measure process (control) & 3.2\\
         
        $\bar{\eta}_k$   & $\mathbb{R}^{m}$ &  Discrete    & Change of measure process (control) & 3.2\\
         
        $W^{\bar\gamma}_s$ & $\mathbb{R}^{d}$ & Continuous & $\mathcal{F}_s$-Wiener process $W^{\bar\gamma}_s := W_s - \int_0^s \bar\gamma_u\, du,$ & 3.2\\

        $v^{\bar\eta}_k$ & $\mathbb{R}^{m}$ & Continuous   & Exploration noise: $v^{\bar\eta}_k \sim \mathcal{N} \left(\bar\eta_k, \Psi_k\right)$ & 3.2\\

        \\
        
        \hline\\

        \multicolumn{5}{l}{\underline{Value function for the stochastic differential game:}} \\
        $u_0(x_0)$  & $\mathbb{R}$    & Discrete   & Value of the stochastic differential game: & 3.3 \\
        & & & {\small $u_0(x_0) := \inf_{\bar h \in \bar{\mathcal A}^{H}_{\mathrm{expl}}}\sup_{\bar\gamma,\bar\eta}
                \mathbf{E}^{\bar\gamma, \bar\eta} \Bigg[  \theta \sum_{k=0}^{K-1} g(X_k, \bar h_k, \bar\eta_k, \bar \gamma_k)\Delta t \Bigg]$} & \\
        $u_k(X_k)$ & $\mathbb{R}$    & Discrete    & Recursive quadratic expression:  & 3.3\\
        & & & $u_k(X_k) = \inf_{\bar h \in \bar{\mathcal A}^{H}_{\mathrm{expl}}}\sup_{\bar\gamma,\bar\eta}  \left\{
        \theta g(X_k, \bar h_k, \bar\eta_k, \bar \gamma_k)\Delta t
        + \mathbf{E}_{k,X_k}^{\bar\gamma, \bar\eta} \left[ V_{k+1}(X_{k+1}) \right]\right\}$ & 3.3\\
        & & & $\phantom{u_k(X_k) :} = X_k' P_k X_k + X_k' p_k + r_k$ & 3.3\\

        \\
        
        \hline\\

        \multicolumn{5}{l}{\underline{Shorthand notation}} \\
        $\mathcal{A}_{k+1}$ &   $\mathbb{R}^d$ & Discrete & $\mathcal{A}_{k+1} = 
        \Lambda' P_{k+1} \Lambda \Delta t - I_d$ & 3.3 \\
        $\mathcal{B}_{k+1}$ & $\mathbb{R}^d$ & Discrete & $\mathcal{B}_{k+1} := 
        \theta \Sigma'\left(\Sigma\Sigma'\right)^{-1}\Sigma
        - \mathcal{A}_{k+1}$ & 3.3
        \\
        $\mathcal{C}_{k+1}\!(\theta)$ & $\mathbb{R}^m$ & Discrete & $\mathcal{C}_{k+1}\!(\theta) = \frac{1}{\theta+1}\Sigma\left[
        I_d
        - \theta \mathcal{A}_{k+1}^{-1} 
        \right]\Sigma'$
        \\
         \hline
    \end{tabular}}
    \caption{Key variables and parameters appearing in Sections 2 and 3}
    \label{tab:variables:summary}
\end{table}

\subsection{Exploratory Controls and the Log-Relative Return Process}\label{sec:sol:exploratorycontrols}

We define the class of exploratory \emph{strategies} and derive a policy-averaged representation of the terminal log-relative return that explicitly reveals the effect of exploration.

\begin{definition}[Class of admissible exploratory strategies $\mathcal{A}^{H}_{\mathrm{expl}}$]\label{def:class:AH:expl}
A strategy $H = (h_s)_{s \in [0,T]}$ with values in $\mathbb{R}^m$  
belongs to the exploratory admissible class $\mathcal{A}^{H}_{\mathrm{expl}}$ if the control process $h_s$ is $\mathcal{F}_s$-measurable and satisfies the following conditions:

\begin{enumerate}[(i)]
    \item $(h_s)_{s \in [0,T]}$ is a \emph{randomized control} according to Definition \ref{def:randomizedcontrol}.
    
    \item \emph{Square-integrability:}  
    \begin{align}
        \mathbb{E}[\|h_k\|^2] = \mathbb{E} [\|\bar h_k\|^2] + \mathrm{tr}(\Psi_k) < +\infty, 
        \quad k=0,\dots,K-1,
    \end{align}
    where $\|h_k\|^2 := h_k'h_k$.
\end{enumerate}

\end{definition}

For our finite-horizon discrete-time representation used below, it is convenient to work on a canonical product-space representation and write the collection of Brownian increments and exploration shocks. We denote elements $\omega \in \Omega$ of the underlying probability space as
\begin{align*}
\omega = (\omega^w, \omega^v) := (w_0, \ldots, w_{K-1}, v_0, \ldots, v_{K-1}),
\end{align*}
where $w_k = W_{k+1} - W_k$ are the Brownian increments and $v_k$ are the random perturbations from exploration.

Using the randomized control representation in Definition \ref{def:randomizedcontrol}, and conditioning on the state process and Brownian increments, we average over the Gaussian exploratory policy at each rebalancing date to obtain the following policy-averaged terminal log-relative return, which we denote by $\bar R_T^\pi$:
\begin{align}\label{eq:excess_return:expl}
\bar R^\pi_T 
=& 
    R_0 + \sum_{k=0}^{K-1} \int_{\mathbb{R}^m} 
\Bigg\{ 
    - \frac{1}{2} h_k'\Sigma\Sigma'h_k
    + h_k' a
    + \frac{1}{2}\Xi'\Xi - c
    + (h_k' A - C) X_k
\Bigg\} \pi(dh;\bar h_k) \Delta t
                        \nonumber\\
&
+ \sum_{k=0}^{K-1} \int_{\mathbb{R}^m} (h_k'\Sigma - \Xi') w_k \, \pi(dh;\bar h_k)
\nonumber\\[2mm]
=& R_0 + \sum_{k=0}^{K-1} 
\Bigg\{
- \frac{1}{2} \bar h_k' \Sigma \Sigma' \bar h_k
- \frac{1}{2} \mathrm{tr}(\Psi_k \Sigma \Sigma')
+ \bar h_k' a
+ \frac{1}{2} \Xi'\Xi - c
+ (\bar h_k' A - C) X_k
\Bigg\} \Delta t
\nonumber\\
&+ \sum_{k=0}^{K-1} (\bar h_k' \Sigma - \Xi') w_k.
\end{align}

\begin{remark}[Interpretation of Exploration]
The additive term
\begin{align}
-\frac{1}{2}\mathrm{tr}(\Psi_k\Sigma\Sigma')
\end{align}
lowers expected log-relative performance and can be interpreted as the instantaneous cost of exploration. The trace term $\mathrm{tr}(\Psi_k \Sigma \Sigma')$ shows that exploration is penalized more strongly when the exploratory covariance $\Psi_k$ allocates mass to directions with larger asset return volatility, as encoded by $\Sigma\Sigma'$.
\end{remark}

We define the risk-sensitive benchmarked criteria for the policy-averaged log excess return $(\bar R^\pi_T-R_0)$ as: 
\begin{align} 
J^\pi(H,\theta)
&:= -\frac{1}{\theta} \ln \mathbf{E} \left[ e^{-\theta (\bar R^\pi_T - R_0)} \right]  \label{eq:J:pi} \\
I^\pi(H,\theta)
&:= e^{-\theta J^\pi(H,\theta)}
= \mathbf{E} \left[ e^{-\theta (\bar R^\pi_T-R_0)}\right].
\label{eq:I:pi}
\end{align}
From this point onward, we work with the policy-averaged criterion and write $J(H,\theta)$ and $I(H,\theta)$ for $J^\pi(H,\theta)$ and $I^\pi(H,\theta)$ when no confusion arises.

\subsection{Free Energy-Entropy Duality for the Randomized Risk-Sensitive Investment Management Problem}\label{sec:sol:FEED}

We introduce on the measurable space
$(\Omega,\mathcal{F},(\mathcal{F}_s)_{s\in[0,T]})$
a probability measure $\mathbb{P}^{\bar\gamma,\bar\eta}$,
distinct from $\mathbb{P}$ and parameterised by two processes:

\begin{enumerate}[(i)]
    \item \emph{Duality–driven continuous-time shift:}  
    a piecewise-constant, right-continuous process 
    $\bar\Gamma=(\bar\gamma_s)_{s\in[0,T]} \in \mathbb{R}^d$  
    such that $\bar\gamma_s$ is constant on each interval 
    $[t_k,t_{k+1})$, $k=0,\ldots,K-1$. 

    \item \emph{Exploration–driven discrete mean shift:} discrete-time process 
    $\bar\eta=(\bar\eta_k)_{k=0,\ldots,K-1}\in\mathbb{R}^m$
    defined on the rebalancing dates.
\end{enumerate}


\begin{remark}[Piecewise-constant process $\bar\gamma$]
Since the state $X$ and control $h$ are updated only on the discrete rebalancing grid $\{t_k\}$ and the objective aggregates interval contributions, we restrict attention to piecewise-constant dual drifts $\bar\gamma$. This is the natural class for the discrete-time dynamic programming formulation developed below.
\end{remark}

 We first define admissible dual variables and the corresponding equivalent measures. The discrete exploratory noise $v_k$ is treated through a finite-dimensional Gaussian mean shift.  
Because $W$ and $v$ are independent under $\mathbb{P}$, the Radon–Nikodym derivative factorises into the product of a continuous-time Girsanov term and a discrete-time exponential of a Gaussian term.


\begin{definition}[Admissible duality strategies $\mathcal{A}^{\bar\Gamma}$]
\label{def:class:AGamma}

A piecewise-constant, right-continuous process 
$\bar\Gamma=(\bar\gamma_s)_{s\in[0,T]}$ belongs to $\mathcal A^{\bar\Gamma}$ if,
for each $k=0,\ldots,K-1$, the value $\bar\gamma_k:=\bar\gamma_{t_k}$ is 
$\mathcal F_{t_k}$-measurable, and it satisfies:

\begin{enumerate}[(i)]
    \item \emph{Piecewise-constant:}  
    $\bar\gamma_s$ is right-continuous and constant on each interval $[t_k,t_{k+1})$.

    \item \emph{Square-integrability:}  
    \begin{align*}
    \mathbb{E}\big[ \lVert \bar\gamma_k \rVert^2 \big] < \infty,
    \qquad k=0,\ldots,K-1,
    \end{align*}
    where $\bar\gamma_k := \bar\gamma_{t_k}$.

    \item \emph{Exponential martingale property:}  
    The process
    \begin{align*}
    \mathcal{L}^{\bar\gamma,W}_T
    := 
    \exp\!\left(
        -\frac{1}{2} \int_0^T \lVert \bar\gamma_s\rVert^2 ds
        + \int_0^T \bar\gamma_s' \, dW_s
    \right)
    \end{align*}
    is a $\mathbb{P}$-exponential martingale.
\end{enumerate}
\end{definition}

\begin{definition}[Admissible exploration shifts $\mathcal{A}^{\bar\eta}$]\label{def:class:Aeta}

A discrete-time process $\bar\eta=(\bar\eta_k)_{k=0,\ldots,K-1}\in\mathbb R^m$
belongs to $\mathcal A^{\bar\eta}$ if:
\begin{enumerate}[(i)]
    \item $\bar\eta_k$ is $\mathcal F_{t_k}$-measurable for $k=0,\ldots,K-1$.
    \item \emph{Square-integrability:} 
        \begin{align*}
            \mathbb{E}\big[ \lVert \bar\eta_k\rVert^2 \big] < \infty,
        \qquad k=0,\ldots,K-1.
        \end{align*}
\end{enumerate}
\end{definition}


For $\bar\Gamma\in \mathcal{A}^{\bar\Gamma}$, the continuous-time change of measure is
\begin{align}
    \frac{d\mathbb{P}^{\bar\gamma}}{d\mathbb{P}}\Big|_{\mathcal{F}_T}
= \mathcal{L}^{\bar\gamma,W}_T.
\end{align}

Define the stochastic process
\begin{align}\label{eq:def:Wgamma}
    W^{\bar\gamma}_s 
    := 
    W_s - \int_0^s \bar\gamma_u\, du,
    \qquad s\in[0,T].
\end{align}
Then $W^{\bar\gamma}$ is a Brownian motion under $\mathbb{P}^{\bar\gamma}$. Note that in \eqref{eq:def:Wgamma}, $\bar\gamma$ is piecewise constant, so the integral over $[0,s]$ reduces to a sum of constant contribution over each interval $[t_k,t_{k+1}),k=0,\ldots,K-1$. 


For $\bar\eta\in\mathcal{A}^{\bar\eta}$, define the discrete-time Radon–Nikodym derivative
\begin{align}
    \frac{d\mathbb{P}^{\bar\eta}}{d\mathbb{P}}\Big{|}_{\mathcal{F}_T}
    &:= 
        \prod_{k=0}^{K-1}
        \exp \left\{
            -\frac{1}{2} \bar\eta_k'\Psi_k^{-1}\bar\eta_k  
            + \bar\eta_k' \Psi_k^{-1}v_k
        \right\}
    =: \mathcal{L}^{\bar\eta,v}.
\end{align}
Under $\mathbb{P}^{\bar\eta}$ (and therefore under $\mathbb P^{\bar\gamma,\bar\eta}$, since the additional change of measure acts only on $W$), conditionally on the filtration $\mathcal{F}_{t_k}$, the random variable $v^{\bar\eta}_k$ is Gaussian with mean $\bar\eta_k$ and covariance $\Psi_k$, that is,  
$v^{\bar\eta}_k \sim \mathcal{N}(\bar\eta_k,\Psi_k)$, and with $v^{\bar\eta}_j$ independent of $v^{\bar\eta}_k$ for $j \neq k$.


Since $W$ and $v$ are independent under $\mathbb{P}$, the joint Radon–Nikodym derivative is
\begin{align}
    \frac{d\mathbb{P}^{\bar\gamma,\bar\eta}}{d\mathbb{P}} \Big|_{\mathcal{F}_T}
    = 
    \mathcal{L}^{\bar\gamma,W}_T \,\mathcal{L}^{\bar\eta,v}
    =:
    \mathcal{L}^{\bar\gamma,\bar\eta}_T.
\end{align}

Definition \ref{def:randomizedcontrol} establishes the randomized exploratory control $(h_s)_{s \in [0,T]}$ as $h_k = \bar h_k + v_k$ under the measure $\mathbb{P}$. Following the change of measure to $\mathbb{P}^{\bar\gamma,\bar\eta}$, we now have $h_k = \bar h_k + v_k^{\bar\eta}$, which therefore depends on the process $\bar\eta$ via the noise term $v^{\bar\eta}_k \sim \mathcal{N} \left(\bar\eta_k, \Psi_k\right)$. Equivalently, the randomized control can  be represented by the measure 
$\pi\left(dh;\bar h_k, \bar\eta_k\right) \sim \mathcal{N} \left(\bar h_k+\bar\eta_k, \Psi_k\right)$.

Under $\mathbb{P}^{\bar\gamma, \bar\eta}$, the state dynamics at \eqref{eq:state} becomes
\begin{align}\label{eq:state:Pbar:rand}
    X_{k+1} = b \Delta t + \tilde{B} X_k +\Lambda \bar\gamma_k \Delta t + \Lambda w^{\bar\gamma}_k,
\end{align}
for $k=0, \ldots, K-1$, and where $w^{\bar\gamma}_k = W^{\bar\gamma}_{k+1} - W^{\bar\gamma}_{k}$, using the shorthand notation $W^{\bar\gamma}_{k} = W^{\bar\gamma}_{t_k}$. 

\begin{remark}[Scaling of exploratory shifts]
In this paper, we model the exploratory perturbations $v^{\bar\eta}_k$ as instantaneous shocks with $$v^{\bar\eta}_k\sim\mathcal{N}(\bar\eta_k,\Psi_k),$$ that occur before each rebalancing date and are independent of the rebalancing interval length $\Delta t$. Equivalently, $\bar\eta_k$ and $\Psi_k$ represent the mean and covariance of the instantaneous exploratory shock applied at time $t_k$.
\end{remark}

We then compute the \emph{relative entropy} under the product change of measure. We express the relative entropy of $\mathbb{P}^{\bar\gamma,\bar\eta}$ with respect to $\mathbb{P}$ as
\begin{align}\label{eq:relativentropy}
D_{\rm KL}(\mathbb{P}^{\bar\gamma,\bar\eta}\|\mathbb{P})
= \mathbf{E}^{\bar\gamma, \bar\eta} \left[ \ln \left( \frac{d\mathbb{P}^{\bar\gamma, \bar\eta}}{d\mathbb{P}} \right) \right]
= \mathbf{E}^{\bar\gamma,\bar\eta}[\ln \mathcal{L}^{\bar\gamma,W}_T]
+ \mathbf{E}^{\bar\gamma,\bar\eta}[\ln \mathcal{L}^{\bar\eta,v}].
\end{align}

Using Girsanov's theorem, the continuous-time part reduces to
\begin{align}
\mathbf{E}^{\bar\gamma,\bar\eta}[\ln \mathcal{L}^{\bar\gamma,W}_T]
&= \mathbf{E}^{\bar\gamma, \bar\eta} \left[ -\frac{1}{2} \int_0^T \lVert \bar\gamma_s\rVert^2 ds
        + \int_0^T \bar\gamma_s' \, dW_s \right]
                                \nonumber\\
&= \mathbf{E}^{\bar\gamma, \bar\eta} \left[ -\frac{1}{2} \int_0^T \lVert \bar\gamma_s\rVert^2 ds
        + \int_0^T \bar\gamma_s' \left( dW^{\bar\gamma}_s 
        + \bar\gamma_s ds\right) \right]
                                \nonumber\\
&= \mathbf{E}^{\bar\gamma, \bar\eta} \left[ \frac{1}{2} \int_0^T \lVert \bar\gamma_s\rVert^2 ds
        + \int_0^T \bar\gamma_s' \, dW^{\bar\gamma}_s \right]
                                \nonumber\\
&= \frac{1}{2} \sum_{k=0}^{K-1}  \mathbf{E}^{\bar\gamma, \bar\eta} \left[ \|\bar\gamma_k\|^2 \Delta t\right],
\end{align}
where we used the fact that $\bar\gamma$ is piecewise constant on $[t_k,t_{k+1}), k=0,\ldots,K-1$, and that $W^{\bar\gamma}$ is a $\mathbb{P}^{\bar\gamma}$-martingale and therefore a $\mathbb{P}^{\bar\gamma, \bar\eta}$-martingale. Notice that since $\frac{d\mathbb{P}^{\bar\gamma,\bar\eta}}{d\mathbb{P}^{\bar\gamma}}\big|_{\mathcal{F}_T}
=\mathcal{L}^{\bar\eta,v}_T$ is independent of $W^{\bar\gamma}$,
the change from $\mathbb{P}^{\bar\gamma}$ to $\mathbb{P}^{\bar\gamma,\bar\eta}$ does not affect the dynamics of $W^{\bar\gamma}$. Therefore $W^{\bar\gamma}$, which is already a $\mathbb{P}^{\bar\gamma}$-martingale, is also a $\mathbb{P}^{\bar\gamma,\bar\eta}$-martingale.

For the discrete-time exploration part,
\begin{align}
\mathbf{E}^{\bar\gamma,\bar\eta}[\ln \mathcal{L}^{\bar\eta,v}]
&= \sum_{k=0}^{K-1} \mathbf{E}^{\bar\gamma, \bar\eta} \left[  
        -\frac{1}{2} \bar\eta_k'\Psi_k^{-1}\bar\eta_k  
        + \bar\eta_k' \Psi_k^{-1}v_k \right]
                        \nonumber\\
&= \frac{1}{2} \sum_{k=0}^{K-1} \mathbf{E}^{\bar\gamma, \bar\eta} \left[ \bar\eta_k' \Psi_k^{-1} \bar\eta_k \right],
\end{align}
where, since $\bar \eta_k$ is $\mathcal{F}_{t_k}$-measurable and $v_k$ has conditional mean $\bar\eta_k$ under $\mathbb{P}^{\bar\gamma,\bar\eta}$, we have $\mathbf{E}^{\bar\gamma, \bar\eta} \left[\bar\eta_k' \Psi_k^{-1}v_k \right] = \mathbf{E}^{\bar\gamma, \bar\eta} \left[\bar\eta_k' \Psi_k^{-1} \bar \eta_k \right]$. 

Hence, the total relative entropy is
\begin{align}\label{eq:relativentropy:total}
D_{\rm KL}(\mathbb{P}^{\bar\gamma,\bar\eta}\|\mathbb{P})
= \frac{1}{2} \sum_{k=0}^{K-1} \mathbf{E}^{\bar\gamma, \bar\eta} \Big[ \, \|\bar\gamma_k\|^2 \Delta t + \bar\eta_k' \Psi_k^{-1} \bar\eta_k \Big].
\end{align}

\begin{remark}
The role of $\bar \gamma$ is to steer Brownian increments in the dual (adversarial) direction, whereas $\bar \eta$ biases the exploratory noise distribution. The resulting relative entropy penalty therefore decomposes into a continuous-time quadratic penalty and a discrete-time Mahalanobis penalty.
\end{remark}

Next, we rewrite the policy-averaged log-relative return under the $\mathbb{P}^{\bar\gamma, \bar\eta}$
\begin{align}\label{eq:reward:rand:Pbar}
\bar R^\pi_T - R_0
=&  \sum_{k=0}^{K-1} \Bigg\{ \int_{\mathbb{R}^m} \left[ 
    \left(- \frac{1}{2} h_k'\Sigma\Sigma'h_k 
    + h_k'a 
    + \frac{1}{2}\Xi'\Xi- c \right)  
    + \left(h_k' A - C \right)X_k \right] \Delta t
                                				\nonumber\\
    &+ \left(h_k'\Sigma - \Xi' \right) \left( \bar\gamma_k \Delta t + w^{\bar\gamma}_k\right)
    \Bigg\} \pi\left(dh;\bar h_k, \bar\eta_k\right)
    							\nonumber\\ 
=&  \sum_{k=0}^{K-1} \int_{\mathbb{R}^m} \Bigg\{ 
    \left(- \frac{1}{2} h_k'\Sigma\Sigma'h_k 
    + h_k'a 
    + \frac{1}{2}\Xi'\Xi- c \right)  
    + \left(h_k' A - C \right)X_k 
                                            \nonumber\\
    &+ \left(h_k'\Sigma - \Xi' \right)\bar\gamma_k
    \Bigg\} \pi\left(dh;\bar h_k, \bar\eta_k\right) \Delta t
    + \sum_{k=0}^{K-1} \int_{\mathbb{R}^m} \left\{ \left(h_k'\Sigma - \Xi' \right)w^{\bar\gamma}_k
    \right\} \pi\left(dh;\bar h_k, \bar\eta_k\right)
    							\nonumber\\ 
    =& \sum_{k=0}^{K-1}  \Bigg\{
    \left(- \frac{1}{2} \left(\bar h_k +\bar\eta_k\right)'\Sigma\Sigma'\left(\bar h_k +\bar\eta_k\right)
     - \frac{1}{2} \tr\left( \Psi_k\Sigma\Sigma'\right) 
    + \left(\bar h_k +\bar\eta_k\right)'a 
    + \frac{1}{2}\Xi'\Xi- c 
    \right)  
                                \nonumber\\
    &+ \left(\left(\bar h_k +\bar\eta_k\right)' A - C \right)X_k 
    + \left((\bar h_k + \bar\eta_k)'\Sigma - \Xi' \right)\bar\gamma_k
    \Bigg\} \Delta t
    + \sum_{k=0}^{K-1} \left((\bar h_k + \bar\eta_k)'\Sigma - \Xi' \right) w^{\bar\gamma}_k.
\end{align}

In the last equality, we used the Gaussian moments under 
$\pi(dh;\bar h_k,\bar\eta_k)\sim\mathcal N(\bar h_k+\bar\eta_k,\Psi_k)$:
\begin{align}
\int_{\mathbb R^m} h\,\pi(dh;\bar h_k,\bar\eta_k)=\bar h_k+\bar\eta_k,
\end{align}
and
\begin{align}
\int_{\mathbb R^m} h'\Sigma\Sigma' h\,\pi(dh;\bar h_k,\bar\eta_k)
=(\bar h_k+\bar\eta_k)'\Sigma\Sigma'(\bar h_k+\bar\eta_k)+\tr(\Psi_k\Sigma\Sigma').
\end{align}
Moreover, conditional on $\mathcal{F}_{t_k}$, the Brownian increment $w^{\bar \gamma}_k$ is independent of the exploratory shock $v^{\bar \eta}_k$, or equivalently, of $h_k$ under the measure $\mathbb{P}^{\bar \gamma, \bar \eta}$. Hence $w_k^{\bar\gamma}$ can be treated as constant with respect to the integration in $h$.

\begin{remark}[Integrability]
Under the Gaussian state dynamics and the square-integrability conditions in Definitions
\ref{def:class:AH:expl}, \ref{def:class:AGamma}, and \ref{def:class:Aeta}, all expectations appearing below are well-defined. In particular, the quadratic and linear terms entering the policy-averaged log-relative return and the entropy penalties are integrable on the finite horizon $[0,T]$.
\end{remark}

Finally, we apply free energy-entropy duality to represent the risk-sensitive control problem as a stochastic game. Applying the definition of free energy at \eqref{eq:free_energy} to $\psi := -\theta (\bar R^\pi_T-R_0)$, we have
\begin{align}
  \mathcal{E}^{\mathbb{P}}\{-\theta (\bar R^\pi_T-R_0)\} = \ln I(H,\theta),
\end{align}
where $I$ is the risk-sensitive criterion defined at \eqref{eq:I:pi}. 

Recalling also the formula for the total entropy at \eqref{eq:relativentropy:total}, the \emph{free energy-entropy duality} from Proposition \ref{prop:FEEDuality} then leads to
\begin{align}\label{eq:criterion:I:Pbar:step1}
	& 	\ln I(H,\theta)
	 								\nonumber\\
	=&	\sup_{\bar \gamma \in \mathcal{A}^{\bar \Gamma}, \bar\eta \in \mathcal{A}^{\bar \eta}}
    		\mathbf{E}^{\bar\gamma, \bar\eta} \left[ -\theta (\bar R^\pi_T-R_0)  
    		-\frac{1}{2} \sum_{k=0}^{K-1} \left( \|\bar\gamma_k\|^2 \Delta t + \bar\eta_k' \Psi_k^{-1} \bar\eta_k \right)
   	 \right]
\end{align}

At this point, we make the following assumption:
\begin{assumption}\label{as:Gamma:Eta}
\begin{align}\label{eq:RS:duality:I:restriction}
    & \sup_{\mathbb{P}^{\gamma, \eta}} \mathbf{E}^{\bar\gamma, \bar\eta} \left[ -\theta (\bar R^\pi_T-R_0)  
    		-\frac{1}{2} \sum_{k=0}^{K-1} \left( \|\bar\gamma_k\|^2 \Delta t + \bar\eta_k' \Psi_k^{-1} \bar\eta_k \right)
   	 \right]
                                \nonumber\\
    &=\sup_{\bar \gamma \in \mathcal{A}^{\bar \Gamma}, \bar\eta \in \mathcal{A}^{\bar \eta}} \mathbf{E}^{\bar\gamma, \bar\eta} \left[ -\theta (\bar R^\pi_T-R_0)  
    		-\frac{1}{2} \sum_{k=0}^{K-1} \left( \|\bar\gamma_k\|^2 \Delta t + \bar\eta_k' \Psi_k^{-1} \bar\eta_k \right)
   	 \right].
\end{align}
\end{assumption}
Theorem \ref{theo:main:recursions} below justifies this restriction by showing that the optimal controls are within the admissible classes $\mathcal A^{\bar\Gamma}$ and $\mathcal A^{\bar\eta}$.

Under Assumption \ref{as:Gamma:Eta},
\begin{align}\label{eq:criterion:I:Pbar:step2}
	& 	\ln I(H,\theta)
	 								\nonumber\\
	=&	\sup_{\bar \gamma \in \mathcal{A}^{\bar \Gamma}, \bar\eta \in \mathcal{A}^{\bar \eta}}
    		\mathbf{E}^{\bar\gamma, \bar\eta} \left[ -\theta (\bar R^\pi_T-R_0)  
    		-\frac{1}{2} \sum_{k=0}^{K-1} \left( \|\bar\gamma_k\|^2 \Delta t + \bar\eta_k' \Psi_k^{-1} \bar\eta_k \right)
   	 \right]
	 								\nonumber\\
	=&	\sup_{\bar \gamma \in \mathcal{A}^{\bar \Gamma}, \bar\eta \in \mathcal{A}^{\bar \eta}}
    		\mathbf{E}^{\bar\gamma, \bar\eta} \Bigg[ -\theta \sum_{k=0}^{K-1}  \Bigg\{
                \left(- \frac{1}{2} \left(\bar h_k +\bar\eta_k\right)'\Sigma\Sigma'\left(\bar h_k +\bar\eta_k\right)
                - \frac{1}{2} \tr\left( \Psi_k\Sigma\Sigma'\right) 
                + \left(\bar h_k +\bar\eta_k\right)'a 
                + \frac{1}{2}\Xi'\Xi- c 
                \right)  
                                \nonumber\\
    &           + \left(\left(\bar h_k +\bar\eta_k\right)' A - C \right)X_k 
                + \left((\bar h_k + \bar\eta_k)'\Sigma - \Xi' \right)\bar\gamma_k
                \Bigg\} \Delta t
            + \sum_{k=0}^{K-1} \left((\bar h_k + \bar\eta_k)'\Sigma - \Xi' \right) w^{\bar\gamma}_k  
                                \nonumber\\
    &           
            -\frac{1}{2} \sum_{k=0}^{K-1} \left( \|\bar\gamma_k\|^2 \Delta t + \bar\eta_k' \Psi_k^{-1} \bar\eta_k \right)
   	        \Bigg]
	 								\nonumber\\
	=&	\sup_{\bar \gamma \in \mathcal{A}^{\bar \Gamma}, \bar\eta \in \mathcal{A}^{\bar \eta}} \left\{
    		\mathbf{E}^{\bar\gamma, \bar\eta} \Bigg[ -\theta \sum_{k=0}^{K-1}  \Bigg\{
                \left(- \frac{1}{2} \left(\bar h_k +\bar\eta_k\right)'\Sigma\Sigma'\left(\bar h_k +\bar\eta_k\right)
                - \frac{1}{2} \tr\left( \Psi_k\Sigma\Sigma'\right) 
                + \left(\bar h_k +\bar\eta_k\right)'a 
                + \frac{1}{2}\Xi'\Xi- c 
                \right) 
            \right.
                                \nonumber\\
    &       \left.
                + \left(\left(\bar h_k +\bar\eta_k\right)' A - C \right)X_k 
                + \left((\bar h_k + \bar\eta_k)'\Sigma - \Xi' \right)\bar\gamma_k
                + \frac{1}{2\theta} \|\bar\gamma_k\|^2 
                + \frac{1}{2\theta\Delta t} \bar\eta_k' \Psi_k^{-1} \bar\eta_k
                \Bigg\} \Delta t 
   	        \Bigg]
            \right.
                                \nonumber\\
    &       \left.    
            -\theta \sum_{k=0}^{K-1} \underbrace{\mathbf{E}^{\bar\gamma, \bar\eta} \Bigg[ 
                 \left((\bar h_k + \bar\eta_k)'\Sigma - \Xi' \right) w^{\bar\gamma}_k            
            \Bigg]}_{=0}
            \right\}
	 								\nonumber\\
    	=&	\sup_{\bar \gamma \in \mathcal{A}^{\bar \Gamma}, \bar\eta \in \mathcal{A}^{\bar \eta}}
    		\theta \sum_{k=0}^{K-1} \mathbf{E}^{\bar\gamma, \bar\eta} \Bigg[   \Bigg\{
                \left(\frac{1}{2} \left(\bar h_k +\bar\eta_k\right)'\Sigma\Sigma'\left(\bar h_k +\bar\eta_k\right)
                + \frac{1}{2} \tr\left( \Psi_k\Sigma\Sigma'\right) 
                - \left(\bar h_k +\bar\eta_k\right)'a 
                - \frac{1}{2}\Xi'\Xi
                + c 
                \right) 
                                    \nonumber\\
    &       
                - \left(\left(\bar h_k +\bar\eta_k\right)' A - C \right)X_k 
                - \left((\bar h_k + \bar\eta_k)'\Sigma - \Xi' \right)\bar\gamma_k
                - \frac{1}{2\theta} \|\bar\gamma_k\|^2 
                - \frac{1}{2\theta\Delta t} \bar\eta_k' \Psi_k^{-1} \bar\eta_k
                \Bigg\} \Delta t 
   	        \Bigg]
	 								\nonumber\\
	=&	\sup_{\bar \gamma \in \mathcal{A}^{\bar \Gamma}, \bar\eta \in \mathcal{A}^{\bar \eta}} 
    		 \mathbf{E}^{\bar\gamma, \bar\eta} \Bigg[  \theta \sum_{k=0}^{K-1} g(X_k, \bar h_k, \bar\eta_k, \bar \gamma_k)\Delta t \Bigg]
\end{align}
where we defined the function $g$ as
\begin{align}\label{eq:g}
	& g(X_k, \bar h_k, \bar\eta_k, \bar\gamma_k)
									\nonumber\\
	=&  \left(\frac{1}{2} \left(\bar h_k          +\bar\eta_k\right)'\Sigma\Sigma'\left(\bar h_k +\bar\eta_k\right)
        + \frac{1}{2} \tr\left( \Psi_k\Sigma\Sigma'\right) 
        - \left(\bar h_k +\bar\eta_k\right)'a 
        - \frac{1}{2}\Xi'\Xi
        + c \right) 
                                    \nonumber\\
    &       
        - \left(\left(\bar h_k +\bar\eta_k\right)' A - C \right)X_k 
        - \left((\bar h_k + \bar\eta_k)'\Sigma - \Xi' \right)\bar\gamma_k
        - \frac{1}{2\theta} \|\bar\gamma_k\|^2 
        - \frac{1}{2\theta\Delta t} \bar\eta_k' \Psi_k^{-1} \bar\eta_k,
\end{align}
and where we used the fact that since $\bar h_k$,$\bar\eta_k$,$\bar\gamma_k$ and $X_k$ are all $\mathcal F_{t_k}$-measurable, and
$w_k^{\bar\gamma}$ is a Brownian increment independent of $\mathcal F_{t_k}$ under
$\mathbb P^{\bar\gamma,\bar\eta}$, we have
$\mathbf E^{\bar\gamma,\bar\eta}\left[
\left((\bar h_k+\bar\eta_k)'\Sigma-\Xi'\right) w_k^{\bar\gamma}
\right]=0$.

We can thus interpret the energy-entropy duality's $\inf\sup$ as a two-player game against Nature. The agent applies control $\bar h$ to minimize the expectation while Nature (via the duality) applies control $\bar\nu := \begin{pmatrix} \bar\gamma' &\bar\eta'\end{pmatrix}'$ to maximize it. The control $\bar \gamma$ steers market noise while the control $\bar \eta$ biases exploration.

We denote by
\[
\bar{\mathcal A}^{H}_{\mathrm{expl}}
:= \left\{ \bar h=(\bar h_k)_{k=0,\ldots,K-1} \;:\; \exists\, H=(h_s)_{s\in[0,T]} \in \mathcal A^{H}_{\mathrm{expl}}
\text{ such that } h_k=\bar h_k+v_k,\; k=0,\ldots,K-1 \right\}
\]
the set of baseline controls induced by admissible exploratory strategies. Taking the infimum over $\bar h$ to minimize the logarithm of the criterion $I$, and with a slight abuse of notation write $I(\bar h,\theta)$ for $I(H,\theta)$ when $H$ is any admissible exploratory strategy inducing $\bar h$,
we have 
\begin{align}\label{eq:EEDuality:inf} 
    & \inf_{\bar h \in \bar{\mathcal A}^{H}_{\mathrm{expl}}} \ln I(\bar h,\theta) 
    = \inf_{\bar h \in \bar{\mathcal A}^{H}_{\mathrm{expl}}} \sup_{\bar\gamma \in \mathcal{A}^{\bar \Gamma},\bar\eta \in \mathcal{A}^{\bar \eta}} \mathbf{E}^{\bar\gamma, \bar\eta} \Bigg[ \theta \sum_{k=0}^{K-1} g(X_k, \bar h_k, \bar\eta_k, \bar \gamma_k)\Delta t \Bigg] 
            \nonumber\\ 
    \Leftrightarrow&
    \ln \inf_{\bar h \in \bar{\mathcal A}^{H}_{\mathrm{expl}}} I(\bar h,\theta)
    = \inf_{\bar h \in \bar{\mathcal A}^{H}_{\mathrm{expl}}} \sup_{\bar\gamma \in \mathcal{A}^{\bar \Gamma},\bar\eta \in \mathcal{A}^{\bar \eta}} \mathbf{E}^{\bar\gamma, \bar\eta} \Bigg[ \theta \sum_{k=0}^{K-1} g(X_k, \bar h_k, \bar\eta_k, \bar \gamma_k)\Delta t \Bigg] 
            \nonumber\\ 
    \Leftrightarrow&
    \inf_{\bar h \in \bar{\mathcal A}^{H}_{\mathrm{expl}}} I(\bar h,\theta)
    = \exp\left\{ \inf_{\bar h \in \bar{\mathcal A}^{H}_{\mathrm{expl}}} \sup_{\bar\gamma \in \mathcal{A}^{\bar \Gamma},\bar\eta \in \mathcal{A}^{\bar \eta}} \mathbf{E}^{\bar\gamma, \bar\eta} \Bigg[ \theta \sum_{k=0}^{K-1} g(X_k, \bar h_k, \bar\eta_k, \bar \gamma_k)\Delta t \Bigg] \right\}, 
\end{align} 
and where the first equivalence follows from Lemma 5.3.1 in \citet{Meneghini1994}.

Focusing on the term inside the exponential on the right-hand side of \eqref{eq:EEDuality:inf}, we consider the game with optimal value  
\begin{align}\label{eq:criterion:I:Pbar}
    u(T;\theta) 
    := \inf_{\bar h \in \bar{\mathcal A}^{H}_{\mathrm{expl}}}\sup_{\bar\gamma \in \mathcal{A}^{\bar \Gamma},\bar\eta \in \mathcal{A}^{\bar \eta}}
    \mathbf{E}^{\bar\gamma, \bar\eta} \Bigg[  \theta \sum_{k=0}^{K-1} g(X_k, \bar h_k, \bar\eta_k, \bar \gamma_k)\Delta t \Bigg]
\end{align}
so that
\begin{align}\label{criterion1}
 \inf_{\bar h \in \bar{\mathcal A}^{H}_{\mathrm{expl}}}I(\bar h,\theta)=\exp\left\{u(T;\theta)\right\}
 \end{align}

\subsection{Saddle Point Representation and Optimal Controls for the Penalized Stochastic Game}

We now solve the stochastic game with optimal value $u(T;\theta)$ as in \eqref{eq:criterion:I:Pbar}. Let $u_k(X_k)$ be the optimal value of the game at the generic time $t_k, \quad k=0,\ldots,K-1$ when the state takes the value $X_k$, and denote $u_0(X_0) := u(T;\theta)$. We apply the Dynamic Programming Principle (DPP) to express the value of the game recursively as:
\begin{align}\label{eq:DPP:TC}
 u_K(X_K)=& 0
\end{align}
and, for $k=K-1,\ldots,0$, 
\begin{align}\label{eq:DPP:DPP} 
u_k(X_k)
    =&\inf_{\bar h_k}\sup_{\bar\gamma_k,\bar\eta_k} \mathbf{E}_{k,X_k}^{\bar\gamma, \bar\eta} \left[
        \theta g(X_k, \bar h_k, \bar\eta_k, \bar \gamma_k)\Delta t
        + u_{k+1}(X_{k+1}) \right]
                                  \nonumber\\ 
    =&\inf_{\bar h_k}\sup_{\bar\gamma_k,\bar\eta_k}  \left\{
        \theta g(X_k, \bar h_k, \bar\eta_k, \bar \gamma_k)\Delta t
        + \mathbf{E}_{k,X_k}^{\bar\gamma, \bar\eta} \left[ u_{k+1}(X_{k+1}) \right]\right\}, 
\end{align}
where $\mathbf{E}_{k,X_k}^{\bar\gamma, \bar\eta}$ denotes the expectation with respect to the measure $\mathbb{P}^{\bar\gamma, \bar\eta}$, conditional on time $t_k$, state process value $X_k$, and with admissible controls at time $t_k$. The function $g$ is defined at \eqref{eq:g}. 

In what follows, we shall show that $u_k(X_k)$ has a quadratic expression in $X_k$  of the form 
\begin{equation}\label{eq:Phi:quadform0}
    u_k(X_k) = \frac{1}{2}X_k' P_k X_k + X_k' p_k + r_k,
\end{equation}
and at the same time, we shall derive the expressions for a stationary point $(\bar h_k^*,\bar\gamma_k^*,\bar\eta_k^*)$ as a candidate saddle point.

On the basis of \eqref{eq:DPP:DPP}, we define the Hamiltonian $\mathcal{H}$
\begin{align}\label{eq:Hamiltonian}
    \mathcal{H}(X_k,\bar h_k,\bar\gamma_k,\bar\eta_k) 
    :=& \theta g(X_k, \bar h_k, \bar\eta_k, \bar \gamma_k)\Delta t
        + \mathbf{E}_{k,X_k}^{\bar\gamma, \bar\eta} \left[ u_{k+1}(X_{k+1}) \right].
\end{align}

Our main result, Theorem \ref{theo:main:recursions}, is preceded by a lemma of independent interest and four propositions.

First we have
\begin{lemma}\label{L1}

Assuming that, at the generic time $t_{k+1}$, the optimal value $u_{k+1}(X_{k+1})$ has a quadratic expression as in \eqref{eq:Phi:quadform0}, for a control triple $(\bar h,\bar\gamma,\bar\eta)$ we have 
\begin{align}\label{eq:Phi:quadform:xt}
    &\mathbf{E}_{k,X_k}^{\bar\gamma, \bar\eta} \left[u_{k+1}(X_{k+1})\right]
                \nonumber\\
    =& \frac{1}{2}\left(b \Delta t + \tilde{B} X_k +\Lambda \bar\gamma_k \Delta t\right)' 
        P_{k+1} 
        \left(b \Delta t + \tilde{B} X_k +\Lambda \bar\gamma_k \Delta t\right) 
                                    \nonumber\\
    &+\frac{1}{2} \tr\left(\Lambda' P_{k+1}\Lambda\right)\Delta t
    + \left(b \Delta t + \tilde{B} X_k +\Lambda \bar\gamma_k \Delta t\right)'p_{k+1}
    + r_{k+1}
\end{align}

\end{lemma}

\begin{proof} We use the state dynamics at \eqref{eq:state:Pbar:rand} to obtain an analytic expression for $\mathbf{E}_{k,X_k}^{\bar\gamma, \bar\eta} \left[u_{k+1}(X_{k+1})\right]$ in terms of $X_k$. Recalling that under $\mathbb{P}^{\bar\gamma, \bar\eta}$, $w^{\bar\gamma}_k = W^{\bar\gamma}_{k+1} - W^{\bar\gamma}_{k}$, we have $\mathbf{E}_{k,X_k}^{\bar\gamma, \bar\eta} \left[w^{\bar\gamma}_k \right]=0$ and $\mathbf{E}_{k,X_k}^{\bar\gamma, \bar\eta} \left[w^{\bar\gamma}_k(w^{\bar\gamma}_k)' \right]= \Delta t \ I_d$, so
\begin{align}
    &   \mathbf{E}_{k,X_k}^{\bar\gamma, \bar\eta} \left[u_{k+1}(X_{k+1})\right]
                                    \nonumber\\
    =&    \mathbf{E}_{k,X_k}^{\bar\gamma, \bar\eta} \left[
        \frac{1}{2}X_{k+1}' P_{k+1} X_{k+1} + X_{k+1}' p_{k+1} + r_{k+1}
    \right]
                                    \nonumber\\
    =&  \mathbf{E}_{k,X_k}^{\bar\gamma, \bar\eta} \left[ 
        \frac{1}{2}\left(b \Delta t + \tilde{B} X_k +\Lambda \bar\gamma_k \Delta t + \Lambda w^{\bar\gamma}_k\right)' 
        P_{k+1} 
        \left(b \Delta t + \tilde{B} X_k +\Lambda \bar\gamma_k \Delta t + \Lambda w^{\bar\gamma}_k\right)
    \right. 
                                    \nonumber\\
    &\left.        
        + \left(b \Delta t + \tilde{B} X_k +\Lambda \bar\gamma_k \Delta t + \Lambda w^{\bar\gamma}_k\right)' p_{k+1} + r_{k+1}
    \right]
                                    \nonumber\\
    =&\frac{1}{2}\left(b \Delta t + \tilde{B} X_k +\Lambda \bar\gamma_k \Delta t\right)' 
        P_{k+1} 
        \left(b \Delta t + \tilde{B} X_k +\Lambda \bar\gamma_k \Delta t\right) 
                                    \nonumber\\
    &\underbrace{+ \left(b \Delta t + \tilde{B} X_k +\Lambda \bar\gamma_k \Delta t\right)' 
        P_{k+1} 
        \Lambda \mathbf{E}_{k,X_k}^{\bar\gamma, \bar\eta} \left[w^{\bar\gamma}_k
        \right]}_{=0}
    +\frac{1}{2} \underbrace{\mathbf{E}_{k,X_k}^{\bar\gamma, \bar\eta} \left[ 
         (w^{\bar\gamma}_k)'\Lambda' P_{k+1} \Lambda w^{\bar\gamma}_k
        \right]}_{= \tr\left(\Lambda' P_{k+1}\Lambda\right)\Delta t}
                                    \nonumber\\
    &+ \left(b \Delta t + \tilde{B} X_k +\Lambda \bar\gamma_k \Delta t\right)'p_{k+1}
    +\underbrace{\mathbf{E}_{k,X_k}^{\bar\gamma, \bar\eta} \left[ \left(w^{\bar\gamma}_k\right)'\right]  \Lambda'p_{k+1}}_{=0}
    + r_{k+1}
                                    \nonumber\\
    =& \frac{1}{2}\left(b \Delta t + \tilde{B} X_k +\Lambda \bar\gamma_k \Delta t\right)' 
        P_{k+1} 
        \left(b \Delta t + \tilde{B} X_k +\Lambda \bar\gamma_k \Delta t\right) 
                                    \nonumber\\
    &+\frac{1}{2} \tr\left(\Lambda' P_{k+1}\Lambda\right)\Delta t
    + \left(b \Delta t + \tilde{B} X_k +\Lambda \bar\gamma_k \Delta t\right)'p_{k+1}
    + r_{k+1},
\end{align}
assuming without loss of generality that $P_{k+1}$ is symmetric\footnote{Since $x' P_{k+1}x = \frac{1}{2} x' (P_{k+1} + P_{k+1}') x $, we may replace $P_{k+1}$ by its symmetric part without changing the quadratic form}. This concludes the proof.
\end{proof}

Next, we derive the saddle point for the game and derive two equivalent representations of the optimal controls. The first representation, in Proposition \ref{prop:Controls}, yields optimality conditions that are easy to check. These conditions are presented in Assumption \ref{as:saddlepoint:cond}. The second representation, in Proposition \ref{prop:Controls:Alt}, produces an optimal asset allocation that is more easily interpretable. We discuss the optimality conditions in Subsection \ref{sec:OptimalityConditions} and interpretability from the perspective of Fractional Kelly Strategies in Subsection \ref{sec:KellyandFKS}. 

For notational convenience, we let
\begin{align}\label{eq:def:CalA}
    \mathcal{A}_{k+1} := 
        \Lambda' P_{k+1} \Lambda \Delta t - I_d,
\end{align}
for $k=0, \ldots, K-1$.
\begin{assumption}\label{as:saddlepoint:cond}
    For $k=0,\ldots,K-1$, assume
    \begin{align}
    - \begin{pmatrix}
        \mathcal{A}_{k+1}    & 0     \\
        0   &  \theta \Sigma\Sigma' \Delta t - \Psi_k^{-1}
    \end{pmatrix} 
    > 0
\end{align}

\end{assumption}
Assumption \ref{as:saddlepoint:cond}, together with Assumption \ref{as:sigma:posdef}, provides sufficient conditions for the existence of a saddle point. We defer the interpretation of this condition until Section \ref{sec:interpretation}.

For a generic control triple $(\bar h, \bar\gamma, \bar\eta)$ and for given $k$, $X$, and a set of model parameters, in what follows, we shall also consider the auxiliary function
\begin{align}\label{eq:auxfunc:F}
    F_k(\bar h, \bar\gamma, \bar\eta ; X)
    :=& \frac{ \theta}{2} \left(\bar h +\bar\eta\right)'\Sigma\Sigma'\left(\bar h +\bar\eta\right)\Delta t
        - \theta\left(\bar h +\bar\eta\right)'(a + A X) \Delta t
       - \theta\left((\bar h + \bar\eta)'\Sigma - \Xi' \right)\bar\gamma \Delta t
                                    \nonumber\\
    &   + \frac{1}{2} \bar\gamma' \mathcal{A}_{k+1} \bar\gamma \Delta t
        + \bar\gamma' \Lambda' \left[ P_{k+1} 
        \left(b \Delta t + \tilde{B} X \right) +  p_{k+1} \right] \Delta t
        - \frac{1}{2} \bar\eta' \Psi_k^{-1} \bar\eta
\end{align}
Under Assumptions \ref{as:sigma:posdef} and \ref{as:saddlepoint:cond}, the search for a saddle point for $u_k(X_k)$ reduces to the search for a saddle point for $F_k(\bar h,\bar\gamma,\bar\eta;X_k)$, as shown in the next proposition.

\begin{proposition}[Saddle Point Representation]\label{prop:SaddlePoint}

If Assumptions \ref{as:sigma:posdef} and \ref{as:saddlepoint:cond} hold, then:

\begin{enumerate}[(i)]
\item The optimal value function $u_k(X_k), k=0,\cdots, K-1$ at \eqref{eq:DPP:DPP} has the following saddle point representation:  
\begin{align}\label{eq:DPPPhiquadform1}
    u_k(X_k) 
    =& \inf_{\bar h_k} \sup_{\bar \gamma_k, \bar \eta_k} \left\{ F_k(\bar h,\bar\gamma,\bar\eta;X_k)
    \right\} 
        + \frac{1}{2}\left(b \Delta t + \tilde{B} X_k\right)' P_{k+1} \left(b \Delta t + \tilde{B} X_k\right)
        + \left(b \Delta t + \tilde{B} X_k \right)'p_{k+1}    
                                    \nonumber\\
    &  + \theta (c + C X_k) \Delta t   
        + \frac{\theta}{2} \tr\left( \Psi_k\Sigma\Sigma'\right)\Delta t 
        - \frac{\theta}{2}\Xi'\Xi \Delta t
        +\frac{1}{2} \tr\left(\Lambda' P_{k+1}\Lambda\right)\Delta t
        + r_{k+1}
\end{align}
where $P_{k+1}, p_{k+1}, r_{k+1}$ are the coefficients of the quadratic form given at \eqref{eq:Phi:quadform0}.

\item The minimax condition for the Hamiltonian defined at \eqref{eq:Hamiltonian} holds:
\begin{align}\label{eq:DPP:minmax}
    \inf_{\bar h_k}\sup_{\bar\gamma_k,\bar\eta_k}
    \mathcal{H}(X_k,\bar h_k,\bar\gamma_k,\bar\eta_k)
    =
    \sup_{\bar\gamma_k,\bar\eta_k}\inf_{\bar h_k}
    \mathcal{H}(X_k,\bar h_k,\bar\gamma_k,\bar\eta_k)
\end{align}

\item We have
\begin{align} \label{eq:saddle:F}
     F_k(h^*_k,\bar\gamma_k,\bar\eta_k;X_k) 
     \leq
     F_k(h^*_k,\gamma^*_k,\eta^*_k;X_k) 
     \leq 
     F_k(\bar h_k,\gamma^*_k,\eta^*_k;X_k).  
\end{align}

\end{enumerate}

\end{proposition}

\begin{proof}
\begin{enumerate}[(i)]

\item Substituting \eqref{eq:g} and \eqref{eq:Phi:quadform:xt} into the dynamic programming equation at \eqref{eq:DPP:DPP}, we obtain
\begin{align}\label{eq:DPPPhiquadform2} 
    u_k(X_k)
    =&\inf_{\bar h_k}\sup_{\bar\gamma_k,\bar\eta_k}  \left\{
        \theta g(X_k, \bar h_k, \bar\eta_k, \bar \gamma_k)\Delta t
        + \mathbf{E}_{k,X_k}^{\bar\gamma, \bar\eta} \left[ u_{k+1}(X_{k+1}) \right]\right\}
                            \nonumber\\
    =& \inf_{\bar h_k}\sup_{\bar\gamma_k,\bar\eta_k}  \Bigg\{
        \frac{ \theta}{2} \left(\bar h_k +\bar\eta_k\right)'\Sigma\Sigma'\left(\bar h_k +\bar\eta_k\right)\Delta t
        + \frac{\theta}{2} \tr\left( \Psi_k\Sigma\Sigma'\right)\Delta t 
        - \theta\left(\bar h_k +\bar\eta_k\right)'a \Delta t
        - \frac{\theta}{2}\Xi'\Xi \Delta t
                                    \nonumber\\
    &       
        + \theta c \Delta t
        - \theta \left(\left(\bar h_k +\bar\eta_k\right)' A - C \right)X_k \Delta t
        - \theta\left((\bar h_k + \bar\eta_k)'\Sigma - \Xi' \right)\bar\gamma_k \Delta t
        - \frac{1}{2} \|\bar\gamma_k\|^2 \Delta t
        - \frac{1}{2} \bar\eta_k' \Psi_k^{-1} \bar\eta_k
                                    \nonumber\\
    &  
        + \frac{1}{2}\left(b \Delta t + \tilde{B} X_k +\Lambda \bar\gamma_k \Delta t\right)' 
        P_{k+1} 
        \left(b \Delta t + \tilde{B} X_k +\Lambda \bar\gamma_k \Delta t\right) 
                                    \nonumber\\
    &   +\frac{1}{2} \tr\left(\Lambda' P_{k+1}\Lambda\right)\Delta t
        + \left(b \Delta t + \tilde{B} X_k +\Lambda \bar\gamma_k \Delta t\right)'p_{k+1}
        + r_{k+1}
        \Bigg\}
                            \nonumber\\
    =& \inf_{\bar h_k}\sup_{\bar\gamma_k,\bar\eta_k}  \Bigg\{
        \frac{ \theta}{2} \left(\bar h_k +\bar\eta_k\right)'\Sigma\Sigma'\left(\bar h_k +\bar\eta_k\right)\Delta t
        + \frac{\theta}{2} \tr\left( \Psi_k\Sigma\Sigma'\right)\Delta t 
        - \theta\left(\bar h_k +\bar\eta_k\right)'a \Delta t
                                    \nonumber\\
    &       
        - \frac{\theta}{2}\Xi'\Xi \Delta t
        + \theta c \Delta t
        - \theta \left(\left(\bar h_k +\bar\eta_k\right)' A - C \right)X_k \Delta t
       - \theta\left((\bar h_k + \bar\eta_k)'\Sigma - \Xi' \right)\bar\gamma_k \Delta t
                                    \nonumber\\
    &       
       - \frac{1}{2} \bar\gamma_k'\bar\gamma_k  \Delta t
       - \frac{1}{2} \bar\eta_k' \Psi_k^{-1} \bar\eta_k
        + \frac{1}{2} \bar\gamma_k' \Lambda' P_{k+1} \Lambda \bar\gamma_k (\Delta t)^2 
        + \bar\gamma_k' \Lambda' P_{k+1} 
        \left(b \Delta t + \tilde{B} X_k \right) \Delta t
                                    \nonumber\\
    &  
        + \frac{1}{2}\left(b \Delta t + \tilde{B} X_k\right)' 
        P_{k+1} 
        \left(b \Delta t + \tilde{B} X_k\right) 
        +\frac{1}{2} \tr\left(\Lambda' P_{k+1}\Lambda\right)\Delta t
        +  \bar\gamma_k' \Lambda' p_{k+1} \Delta t
                                    \nonumber\\
    &  
        + \left(b \Delta t + \tilde{B} X_k \right)'p_{k+1}
        + r_{k+1}
        \Bigg\}
                                    \nonumber\\
    =& \inf_{\bar h_k}\sup_{\bar\gamma_k,\bar\eta_k}  \Bigg\{
        \frac{ \theta}{2} \left(\bar h_k +\bar\eta_k\right)'\Sigma\Sigma'\left(\bar h_k +\bar\eta_k\right)\Delta t
        - \theta\left(\bar h_k +\bar\eta_k\right)'(a + A X_k) \Delta t
                                    \nonumber\\
    &   
        - \theta\left((\bar h_k + \bar\eta_k)'\Sigma - \Xi' \right)\bar\gamma_k \Delta t
        + \frac{1}{2} \bar\gamma_k' \mathcal{A}_{k+1} \bar\gamma_k \Delta t
        + \bar\gamma_k' \Lambda' \left[ P_{k+1} 
        \left(b \Delta t + \tilde{B} X_k \right) +  p_{k+1} \right] \Delta t
                                    \nonumber\\
    &   - \frac{1}{2} \bar\eta_k' \Psi_k^{-1} \bar\eta_k
        \Bigg\}
        + \frac{1}{2}\left(b \Delta t + \tilde{B} X_k\right)' P_{k+1} \left(b \Delta t + \tilde{B} X_k\right)
        + \left(b \Delta t + \tilde{B} X_k \right)'p_{k+1}      
                                    \nonumber\\
    &   + \theta (c + C X_k) \Delta t 
        + \frac{\theta}{2} \tr\left( \Psi_k\Sigma\Sigma'\right)\Delta t 
        - \frac{\theta}{2}\Xi'\Xi \Delta t
        +\frac{1}{2} \tr\left(\Lambda' P_{k+1}\Lambda\right)\Delta t
        + r_{k+1}
\end{align}
from which, using the definition of $F_k$ in \eqref{eq:auxfunc:F}, we obtain the desired representation of $u_k(X_k)$, which proves the first part of the proposition.

\item The function $F_k$ at \eqref{eq:auxfunc:F} is quadratic in the controls $\bar h$, $\bar\gamma$, and $\bar\eta$. Under Assumptions \ref{as:sigma:posdef} and \ref{as:saddlepoint:cond}, $F_k$ is convex in $\bar h$ and concave in $\bar\gamma$ and $\bar\eta$. Moreover, the Hessian is independent of the value of the controls. So $F_k$ admits a unique stationary point $(h^*, \gamma^*, \eta^*)$, which is affine in the state $X_k$. Since $F_k$ is jointly continuous, convex in $\bar h$, concave in $(\bar\gamma, \bar\eta)$, and admits a unique stationary point, that stationary point is the unique saddle point of $F_k$. Therefore the minimax equality holds for the Hamiltonian inside \eqref{eq:DPPPhiquadform2}, that is, 
\begin{align}\label{eq:DPP:minmax:proof}
    \inf_{\bar h_k}\sup_{\bar\gamma_k,\bar\eta_k}
    \mathcal{H}(X_k,\bar h_k,\bar\gamma_k,\bar\eta_k)
    =
    \sup_{\bar\gamma_k,\bar\eta_k}\inf_{\bar h_k}
    \mathcal{H}(X_k,\bar h_k,\bar\gamma_k,\bar\eta_k).
\end{align}

\item It follows from \eqref{eq:DPP:minmax} that we can apply the saddle point condition
\begin{align}\label{eq:saddle:Ham:1}
    \mathcal{H}(X_k,h^*_k,\bar\gamma_k,\bar\eta_k)
\leq \mathcal{H}(X_k,h^*_k,\gamma^*_k,\eta^*_k)
\leq \mathcal{H}(X_k,\bar h_k,\gamma^*_k,\eta^*_k)
\end{align} 
to solve the game. Applying \eqref{eq:DPPPhiquadform1} together with the definition of $F_k$ at \eqref{eq:auxfunc:F} and the definition of supremum/infimum, the saddle point condition \eqref{eq:saddle:Ham:1} simplifies to
\begin{align}
     F_k(h^*_k,\bar\gamma_k,\bar\eta_k;X_k) 
     \leq
     F_k(h^*_k,\gamma^*_k,\eta^*_k;X_k) 
     \leq 
     F_k(\bar h_k,\gamma^*_k,\eta^*_k;X_k).  
\end{align}
Hence, the search for a saddle point for $u_k(X_k)$ reduces to the search for a saddle point for $F_k(\bar h,\bar\gamma,\bar\eta;X_k)$.
\end{enumerate}

\end{proof}

\begin{remark}
    Because $F_k$ is quadratic with constant Hessian and is convex-concave under Assumptions \ref{as:sigma:posdef} and \ref{as:saddlepoint:cond}, the stationary point can be computed by solving the first-order conditions in any order of the variables.
\end{remark}

\begin{proposition}\label{prop:Controls}

Under Assumptions \ref{as:sigma:posdef} and \ref{as:saddlepoint:cond}, the candidate optimal controls are given by the stationary point  $(h^*_k,\gamma^*_k,\eta^*_k), k=0,\ldots K-1$ of the quadratic function
$F_k(\bar h,\bar\gamma,\bar\eta;X_k)$ at \eqref{eq:auxfunc:F} and can be explicitly expressed as:
\begin{align}\label{eq:hstar}
h^*_k
    =&  \left(\Sigma\Sigma'\right)^{-1}\Bigg\{ 
        \left\{
            I_m
            - \theta \Sigma \mathcal{B}_{k+1}^{-1} \Sigma'\left(\Sigma\Sigma'\right)^{-1}
        \right\}(a + A X_k)
        + \Sigma \mathcal{B}_{k+1}^{-1}\left\{  
            \Lambda' P_{k+1} \left(b \Delta t + \tilde{B} X_k \right)
            + \Lambda'   p_{k+1} 
            + \theta \Xi
        \right\}
       \Bigg\}
\end{align}
\begin{align}\label{eq:gammastar}
    \gamma^*_k 
    =&  \mathcal{B}_{k+1}^{-1}\left\{
        - \theta \Sigma'\left(\Sigma\Sigma'\right)^{-1} (a + A X_k) 
        + \Lambda' P_{k+1} \left(b \Delta t + \tilde{B} X_k \right) 
        + \Lambda' p_{k+1}
        + \theta \Xi
    \right\}
\end{align}
\begin{align}\label{eq:etastar} 
    \eta^*_k =& 0.
\end{align}
where
\begin{align}\label{eq:def:CalB}
    \mathcal{B}_{k+1} := 
        \theta \Sigma'\left(\Sigma\Sigma'\right)^{-1}\Sigma
        - \mathcal{A}_{k+1}, 
\end{align}
for $k=0, \ldots, K-1$.

\end{proposition}

\begin{proof}
According to the proof of the previous Proposition, it suffices to look for a saddle point for $F_k(\bar h,\bar\gamma,\bar\eta;X_k)$. By the first order condition, $h^*$ satisfies
\begin{align}\label{eq:hstar:1}
    &\frac{\partial F_k}{\partial \bar h}(\bar h, \bar\gamma, \bar\eta; X_k)\Big|_{\bar h = h^*, \bar \gamma = \gamma^*, \bar\eta = \eta^*}
    = 0                                                                             
                                                        \nonumber\\
    \Leftrightarrow&
    \theta \Sigma\Sigma'\left(h^* + \eta^* \right)\Delta t
        - \theta (a + A X_k) \Delta t
       - \theta \Sigma \gamma^* \Delta t
        = 0
                                                    \nonumber\\
\Leftrightarrow&
    \left(h^* + \eta^*\right)
        =  \left(\Sigma\Sigma'\right)^{-1}\left[ 
            (a + A X_k)
       +  \Sigma \gamma^*
       \right],
\end{align}
Assumption \ref{as:sigma:posdef} ensures that $F_k$ reaches a minimum in its argument $\bar h$ at $h^*$.

Next, we apply the first order condition to $\bar\eta$:
\begin{align}\label{eq:etastar:1}
    &\frac{\partial F_k}{\partial \bar \eta}(\bar h, \bar\gamma, \bar\eta; X_k)\Big|_{\bar h = h^*, \bar \gamma = \gamma^*, \bar\eta = \eta^*}
    = 0                                                                             
                                                        \nonumber\\
    \Leftrightarrow&
    \theta \Sigma\Sigma'\left(h^* + \eta^*\right)\Delta t
        - \theta (a + A X_k) \Delta t
        - \theta \Sigma \gamma^* \Delta t
        - \Psi_k^{-1} \eta^*
        = 0
                                                        \nonumber\\
    \Leftrightarrow&
     \Psi_k^{-1} \eta^*
     = 
     \underbrace{\theta \Sigma\Sigma'\left(h^* + \eta^*\right)\Delta t
        - \theta (a + A X_k) \Delta t
        - \theta \Sigma \gamma^* \Delta t}_{=0 \text{ by } \eqref{eq:hstar:1}}
        = 0
\end{align}
Hence, the candidate exploratory bias $\eta^* = 0$. By Assumption \ref{as:saddlepoint:cond}, $F_k$ reaches a maximum in its argument $\bar \eta$ at $\eta^* = 0$.

Finally, 
\begin{align}\label{eq:gammastar:1}
    &\frac{\partial F_k}{\partial \bar \gamma}(\bar h, \bar\gamma, \bar\eta; X_k)\Big|_{\bar h = h^*, \bar \gamma = \gamma^*, \bar\eta = \eta^*}
    = 0                                                                             
                                                \nonumber\\
    \Leftrightarrow&
        - \theta\left(\Sigma'(h^* + \eta^*) - \Xi \right)\Delta t
        + \mathcal{A}_{k+1} \gamma^* \Delta t
        + \Lambda' \left[ P_{k+1} 
        \left(b \Delta t + \tilde{B} X_k \right) +  p_{k+1} \right]\Delta t
        = 0                        
                                                \nonumber\\
    \Leftrightarrow&
        \mathcal{A}_{k+1} \gamma^*
        - \theta\Sigma'(h^* + \eta^*)
        + \theta \Xi 
        + \Lambda' \left[ P_{k+1} 
        \left(b \Delta t + \tilde{B} X_k \right) +  p_{k+1} \right]
        = 0
\end{align}
By Assumption \ref{as:saddlepoint:cond}, $F_k$ reaches a maximum in its argument $\bar \gamma$ at $\gamma^*$.

From \eqref{eq:hstar:1}, we have 
\begin{align}
    \theta \Sigma'\left(h^* + \eta^*\right)
    = \theta \Sigma'\left(\Sigma\Sigma'\right)^{-1}\left[ 
            (a + A X_k)
       +  \Sigma \gamma^*
       \right]
    ,
\end{align}
thus \eqref{eq:gammastar:1} becomes 
{\small
\begin{align}\label{eq:gammastar:2}
    &   \mathcal{A}_{k+1} \gamma^*
        - \theta \Sigma'\left(\Sigma\Sigma'\right)^{-1}\left[ 
            (a + A X_k)
            +  \Sigma \gamma^*
        \right]
        + \theta \Xi 
        + \Lambda' \left[ P_{k+1} 
        \left(b \Delta t + \tilde{B} X_k \right) +  p_{k+1} \right]
        = 0
                    \nonumber\\
    \Leftrightarrow&
        \left[\Lambda' P_{k+1} \Lambda \Delta t - I_d - \theta \Sigma'\left(\Sigma\Sigma'\right)^{-1}\Sigma
        \right]\gamma^*
        - \theta \Sigma'\left(\Sigma\Sigma'\right)^{-1} (a + A X_k)
        + \theta \Xi 
        + \Lambda' \left[ P_{k+1} 
        \left(b \Delta t + \tilde{B} X_k \right) +  p_{k+1} \right]
        = 0
                    \nonumber\\
    \Leftrightarrow&
        \gamma^* = \left[\Lambda' P_{k+1} \Lambda \Delta t - I_d - \theta \Sigma'\left(\Sigma\Sigma'\right)^{-1}\Sigma
        \right]^{-1}
        \left\{
            \theta \Sigma'\left(\Sigma\Sigma'\right)^{-1} (a + A X_k)
            - \theta \Xi 
            - \Lambda' \left[ P_{k+1} 
            \left(b \Delta t + \tilde{B} X_k \right) +  p_{k+1} \right]
        \right\}
\end{align}
}
Thus, 
\begin{align}\label{eq:gammastar:3}
    \gamma^* = \mathcal{B}_{k+1}^{-1}\left\{
        - \theta \Sigma'\left(\Sigma\Sigma'\right)^{-1} (a + A X_k) 
        + \Lambda' P_{k+1} \left(b \Delta t + \tilde{B} X_k \right) 
        + \Lambda' p_{k+1}
        + \theta \Xi
    \right\},
\end{align}
with $\mathcal{B}_{k+1}$ defined at \eqref{eq:def:CalB}.

Under Assumption \ref{as:saddlepoint:cond}, $F_k$ reaches a maximum in its argument $\bar \gamma$ at $\gamma^*$. Substituting $\eta^*=0$ and \eqref{eq:gammastar:3} into \eqref{eq:hstar:1}, we get:
\begin{align}\label{eq:hstar:2}
    h^*
    =&  \left(\Sigma\Sigma'\right)^{-1}\Bigg\{ 
            (a + A X_k)
       +  \Sigma \mathcal{B}_{k+1}^{-1}\left\{
        - \theta \Sigma'\left(\Sigma\Sigma'\right)^{-1} (a + A X_k)
        + \theta \Xi 
        + \Lambda' \left[ P_{k+1} 
            \left(b \Delta t + \tilde{B} X_k \right) +  p_{k+1} \right]
    \right\}
                    \nonumber\\
    =&  \left(\Sigma\Sigma'\right)^{-1}\Bigg\{ 
        \left\{
            I_m
            - \theta \Sigma \mathcal{B}_{k+1}^{-1} \Sigma'\left(\Sigma\Sigma'\right)^{-1}
        \right\}(a + A X_k)
        + \Sigma \mathcal{B}_{k+1}^{-1}\left\{  
            \Lambda' P_{k+1}\left(b \Delta t + \tilde{B} X_k \right)
            + \Lambda'   p_{k+1} 
            + \theta \Xi
        \right\}
       \Bigg\}
\end{align}

\end{proof}

The following Proposition presents an alternative characterization of the candidate controls which will prove convenient to interpret the optimal asset allocation as a Fractional Kelly Strategy in Section \ref{sec:KellyandFKS} below.

\begin{proposition}\label{prop:Controls:Alt}

Under Assumptions \ref{as:sigma:posdef} and \ref{as:saddlepoint:cond}, the candidate optimal controls are given by the stationary point  $(h^*_k,\gamma^*_k,\eta^*_k), k=0,\ldots K-1$ of the quadratic function
$F_k(\bar h,\bar\gamma,\bar\eta;X_k)$ at \eqref{eq:auxfunc:F} and can be explicitly expressed as:
\begin{align}\label{eq:hstar:alt}
h^*_k
    =&  \frac{1}{\theta+1}\mathcal{C}_{k+1}^{-1}\!(\theta)
    \Bigg[
        (a + A X_k)
        + \theta \mathcal{A}_{k+1}^{-1} \Sigma
        \left\{
            \Xi
            - \Lambda'\left[ 
                P_{k+1} \left(b \Delta t + \tilde{B} X_k \right) 
                +  p_{k+1} 
            \right]
        \right\}
    \Bigg]   
\end{align}
\begin{align}\label{eq:gammastar:alt} 
    \gamma^*_k 
    =&  \mathcal{A}_{k+1}^{-1}
        \Bigg\{ 
            \frac{\theta}{\theta+1} \Sigma'\mathcal{C}_{k+1}^{-1}\!(\theta) (a + A X_k)
        + \theta \left( \frac{\theta}{\theta+1}
                \Sigma' \mathcal{C}_{k+1}^{-1}\!(\theta) 
            \mathcal{A}_{k+1}^{-1} - I_d \right)\Sigma\Xi
                                            \nonumber\\
        & 
            - \left(\frac{\theta}{\theta+1}\Sigma'
                \mathcal{C}_{k+1}^{-1}\!(\theta)
            \mathcal{A}_{k+1}^{-1} \Sigma 
            + I_d \right) 
            \Lambda'
            \left[ 
                    P_{k+1} \left(b \Delta t + \tilde{B} X_k \right) 
                    +  p_{k+1} 
                \right] 
            \Bigg\}.
\end{align}
\begin{align}\label{eq:etastar:alt} 
    \eta^*_k =& 0.
\end{align}
where
\begin{align}\label{eq:CalC}
    \mathcal{C}_{k+1}\!(\theta)
    :=
        \frac{1}{\theta+1}\Sigma\left[
        I_d
        - \theta \mathcal{A}_{k+1}^{-1} 
        \right]\Sigma',
\end{align}
for $k=0, \ldots, K-1$.
\end{proposition}

\begin{proof}
Propositions \ref{prop:SaddlePoint} and \ref{prop:Controls} already showed that $F_k$ is convex in $\bar h$ and concave in $\bar \gamma$ and $\bar \eta$ provided that Assumptions \ref{as:sigma:posdef} and \ref{as:saddlepoint:cond} hold.

To prove Proposition \ref{prop:Controls:Alt}, we proceed in the reverse order of Proposition \ref{prop:Controls}, applying the first order condition with respect to $\bar \gamma$ and then with respect to $\bar h$:

\begin{align}\label{eq:FOC:gammastar}
    &\frac{\partial F_k}{\partial \bar \gamma}(h^*, \bar\gamma, \eta^*; X_k)\Big|_{\bar h = h^*, \bar \gamma = \gamma^*, \bar\eta = \eta^*}
    = 0                                                                             
                                                \nonumber\\
    \Leftrightarrow&
        \mathcal{A}_{k+1} \gamma^*
        - \theta\Sigma'(h^* + \eta^*)
        + \theta \Xi 
        + \Lambda' \left[ P_{k+1} 
        \left(b \Delta t + \tilde{B} X_k \right) +  p_{k+1} \right]
        = 0
                                                \nonumber\\
    \Leftrightarrow&
         \gamma^*
        = \mathcal{A}_{k+1}^{-1}
        \left\{
            \theta\Sigma'(h^* + \eta^*)
            - \theta \Xi 
            - \Lambda' \left[ P_{k+1} 
            \left(b \Delta t + \tilde{B} X_k \right) +  p_{k+1} \right] 
        \right\},
\end{align}
By Assumption \ref{as:saddlepoint:cond}, $\mathcal{A}_{k+1} < 0$, hence it is invertible. The function $F_k$ reaches a maximum in its argument $\bar \gamma$ if $\mathcal{A}_{k+1} <0$. In the following, we also assume that $P_{k+1}$ is symmetric to be able to write $x'P_{k+1}y = y'P_{k+1}x$.

We substitute this expression into $F_k$
\begin{align}
    & F_k(h^*, \gamma^*, \eta^* ; X_k)
                                    \nonumber\\
    =& \frac{ \theta}{2} \left(h^* +\eta^*\right)'\Sigma\Sigma'\left(h^* +\eta^*\right)\Delta t
        - \theta\left(h^* +\eta^*\right)'(a + A X_k) \Delta t
                                    \nonumber\\
    &   - \theta\left((h^* + \eta^*)'\Sigma - \Xi' \right)
        \mathcal{A}_{k+1}^{-1}
        \left\{
            \theta\Sigma'(h^* + \eta^*)
            - \theta \Xi 
            - \Lambda' \left[ P_{k+1} 
            \left(b \Delta t + \tilde{B} X_k \right) +  p_{k+1} \right] 
        \right\} \Delta t
                                    \nonumber\\
    &   + \frac{1}{2} 
        \left\{
            \theta\Sigma'(h^* + \eta^*)
            - \theta \Xi 
            - \Lambda' \left[ P_{k+1} 
            \left(b \Delta t + \tilde{B} X_k \right) +  p_{k+1} \right] 
        \right\}'
                                     \nonumber\\       
    &   \cdot
        \cancel{\mathcal{A}_{k+1}^{-1}}
        \cancel{\mathcal{A}_{k+1}} 
        \mathcal{A}_{k+1}^{-1}
        \left\{
            \theta\Sigma'(h^* + \eta^*)
            - \theta \Xi 
            - \Lambda' \left[ P_{k+1} 
            \left(b \Delta t + \tilde{B} X_k \right) +  p_{k+1} \right] 
        \right\}\Delta t
                                    \nonumber\\
    &   +\left[ P_{k+1} 
        \left(b \Delta t + \tilde{B} X_k \right) +  p_{k+1} \right]' \Lambda
        \mathcal{A}_{k+1}^{-1}
        \left\{
            \theta\Sigma'(h^* + \eta^*)
            - \theta \Xi 
            - \Lambda' \left[ P_{k+1} 
            \left(b \Delta t + \tilde{B} X_k \right) +  p_{k+1} \right] 
        \right\} \Delta t
                                    \nonumber\\
    &   - \frac{1}{2} (\eta^*)' \Psi_k^{-1} \eta^*
                                    \nonumber\\
    =&
        \frac{ \theta}{2} \left(h^* +\eta^*\right)'\Sigma
        \left\{
            I_d
            - \theta \mathcal{A}_{k+1}^{-1} 
        \right\}
        \Sigma'\left(h^* +\eta^*\right)\Delta t
                                    \nonumber\\
    &   - \theta \left(h^* +\eta^*\right)' 
        \Bigg[
            (a + A X_k)
            + \theta \mathcal{A}_{k+1}^{-1} \Sigma
            \left\{
                \Xi
                - \Lambda'\left[ 
                    P_{k+1} \left(b \Delta t + \tilde{B} X_k \right) 
                    +  p_{k+1} 
                \right]
            \right\}           
        \Bigg]\Delta t
    - \frac{\theta^2}{2} \Xi'
        \mathcal{A}_{k+1}^{-1}
        \Xi \Delta t
                                    \nonumber\\
    &   - \frac{1}{2} 
            \left[ 
                P_{k+1}\left(b \Delta t + \tilde{B} X_k \right)    
                +  p_{k+1}
            \right]'\Lambda 
            \mathcal{A}_{k+1}^{-1}
            \Lambda' \left[ 
                P_{k+1}\left(b \Delta t + \tilde{B} X_k \right) 
                +  p_{k+1} 
            \right]\Delta t
    - \frac{1}{2} (\eta^*)' \Psi_k^{-1} \eta^*
\end{align} 

By the first order condition, $h^*$ satisfies
\begin{align}
    &\frac{\partial F_k}{\partial \bar h}(\bar h, \bar\gamma, \bar\eta; X_k)\Big|_{\bar h = h^*, \bar \gamma = \gamma^*, \bar\eta = \eta^*}
    = 0                                                                        \nonumber\\
    \Leftrightarrow&
    \theta \Sigma
        \left\{
            I_d
            - \theta \mathcal{A}_{k+1}^{-1} 
        \right\}
        \Sigma'\left(h^* +\eta^*\right)\Delta t
                                        \nonumber\\
    &   - \theta 
        \Bigg[
            (a + A X_k)
            + \theta \mathcal{A}_{k+1}^{-1} \Sigma
            \left\{
                \Xi
                - \Lambda'\left[ 
                    P_{k+1} \left(b \Delta t + \tilde{B} X_k \right) 
                    +  p_{k+1} 
                \right]
            \right\}           
        \Bigg]\Delta t
        = 0 
\end{align}
and $\eta^*$ satisfies
\begin{align}
    &\frac{\partial F_k}{\partial \bar \eta}(\bar h, \bar\gamma, \bar\eta; X_k)\Big|_{\bar h = h^*, \bar \gamma = \gamma^*, \bar\eta = \eta^*}
    = 0                                                                        \nonumber\\
    \Leftrightarrow&
    \theta \Sigma
        \left\{
            I_d
            - \theta \mathcal{A}_{k+1}^{-1} 
        \right\}
        \Sigma'\left(h^* +\eta^*\right)\Delta t
                                        \nonumber\\
    &   - \theta 
        \Bigg[
            (a + A X_k)
            + \theta \mathcal{A}_{k+1}^{-1} \Sigma
            \left\{
                \Xi
                - \Lambda'\left[ 
                    P_{k+1} \left(b \Delta t + \tilde{B} X_k \right) 
                    +  p_{k+1} 
                \right]
            \right\}           
        \Bigg]\Delta t
        - \Psi_k^{-1}\eta^* 
        = 0 
                                        \nonumber\\
    \Leftrightarrow&
         \frac{\partial F_k}{\partial \bar h}(\bar h, \bar\gamma, \bar\eta; X_k)\Big|_{\bar h = h^*, \bar \gamma = \gamma^*, \bar\eta = \eta^*} - \Psi_k^{-1}\eta^*
    = 0
                                        \nonumber\\
    \Leftrightarrow&
    \Psi_k^{-1}\eta^* = 0
\end{align}
which shows that $\eta^* = 0$, as expected. Here, $\eta^*$ maximizes $F_k$ in its argument $\bar \eta$ if 
\begin{align}\label{eq:SOC:etastar:Alt}
    \theta \Sigma
        \left\{
            I_d
            - \theta \mathcal{A}_{k+1}^{-1} 
        \right\}\Sigma'\Delta t
- \Psi_k^{-1} < 0.
\end{align} 
With $\eta^* = 0$, the first order condition for $\bar h$ becomes 
\begin{align}\label{eq:hstar:2:alt:app}
    &\frac{\partial F_k}{\partial \bar h}(\bar h, \bar\gamma, \bar\eta; X_k)\Big|_{\bar h = h^*, \bar \gamma = \gamma^*, \bar\eta = 0}
    = 0                                                                        \nonumber\\
    \Leftrightarrow&
    h^* = \frac{1}{\theta+1}\mathcal{C}_{k+1}^{-1}\!(\theta)
    \Bigg[
        (a + A X_k)
        + \theta \mathcal{A}_{k+1}^{-1} \Sigma
        \left\{
            \Xi
            - \Lambda'\left[ 
                P_{k+1} \left(b \Delta t + \tilde{B} X_k \right) 
                +  p_{k+1} 
            \right]
        \right\}
    \Bigg] 
\end{align}
with $\mathcal{C}_{k+1}\!(\theta)$ defined at \eqref{eq:CalC}. The function $F_k$ reaches a minimum in its argument $\bar h$ if 
\begin{align}\label{eq:SOC:hstar:Alt}
    (\theta+1) \mathcal{C}_{k+1}\!(\theta)     
    =
    \Sigma\left[
        I_d
        - \theta \left(\Lambda' P_{k+1} \Lambda \Delta t - I_d\right)^{-1} 
        \right]\Sigma' > 0.
\end{align}
This condition ensures strict convexity of $F_k$ in 
$\bar h$, hence uniqueness of $h^*$.

To get \eqref{eq:gammastar:alt}, we substitute the expression we obtained for $h^*$ and $\eta^* = 0$ into \eqref{eq:FOC:gammastar}, concluding the proof.

\end{proof}

Although Proposition \ref{prop:Controls:Alt} is stated under Assumptions \ref{as:sigma:posdef} and \ref{as:saddlepoint:cond}, the derivation naturally reveals the equivalent sufficient conditions \eqref{eq:SOC:hstar:Alt}, \eqref{eq:SOC:etastar:Alt}, and $\mathcal{A}_{k+1}<0$, which are summarized in the remark below.

\begin{remark}
    The conditions at \eqref{eq:SOC:hstar:Alt} and \eqref{eq:SOC:etastar:Alt}, together with the condition $\mathcal{A}_{k+1}<0$, provide a set of sufficient conditions for the existence of a saddle point that is equivalent to the sufficient conditions stated in Assumptions \ref{as:sigma:posdef} and \ref{as:saddlepoint:cond}. Section \ref{sec:interpretation} discusses the intuition and implications of these conditions.
\end{remark}

\subsection{Main Result}

Before proving the main result, we use the first characterization of the controls in Proposition \ref{prop:Controls} to derive a quadratic representation of $F_k$ at the saddle point $(h_k^*,\gamma_k^*,\eta_k^*)$.

\begin{proposition}\label{P2}
Under Assumptions \ref{as:sigma:posdef} and \ref{as:saddlepoint:cond}, at the saddle point $(h_k^*,\gamma_k^*,\eta_k^*)$, the function $F_k(\bar h_k,\bar\gamma_k,\bar\eta_k;X_k)$ from \eqref{eq:auxfunc:F} has the quadratic representation
\begin{align}\label{eq:auxfunc:F:optim:final}
    F_k(h_k^*, \gamma_k^*, \eta_k^*;X_k)  
    =& \frac{1}{2} X_k' \mathfrak{Q}_{k+1} X_k 
        + X_k' \mathfrak{q}_{k+1}  
        + \mathfrak{l}_{k+1},      
\end{align}
with
\begin{align}\label{eq:Qfrak}
\mathfrak{Q}_{k+1}
=&  -\theta A'\left(\Sigma\Sigma'\right)^{-1}A\Delta t
    + \theta^2 
        A'\left(\Sigma\Sigma'\right)^{-1} \Sigma 
        \mathcal{B}_{k+1}^{-1}
        \Sigma'\left(\Sigma\Sigma'\right)^{-1} A 
        \Delta t
                                    \nonumber\\
&   - 2\theta A'\left(\Sigma\Sigma'\right)^{-1}\Sigma   
        \mathcal{B}_{k+1}^{-1}
        \Lambda' P_{k+1} \tilde{B}
        \Delta t   
    + \tilde{B}' P_{k+1} \Lambda 
        \mathcal{B}_{k+1}^{-1}
        \Lambda' P_{k+1} \tilde{B}
        \Delta t,
\end{align}   
\begin{align}\label{eq:qfrak}
\mathfrak{q}_{k+1}
=&  
    - \theta A'\left(\Sigma\Sigma'\right)^{-1} a \Delta t
    + \left\{
        - \theta A'\left(\Sigma\Sigma'\right)^{-1} \Sigma 
        + \tilde{B}' P_{k+1} \Lambda 
    \right\}
    \mathcal{B}_{k+1}^{-1}
                        \nonumber\\
&   \cdot 
        \left\{
        - \theta \Sigma'\left(\Sigma\Sigma'\right)^{-1} a
        + \Lambda' \left[ 
            P_{k+1}b \Delta t +  p_{k+1}  
        \right] 
        + \theta \Xi
    \right\} 
        \Delta t 
\end{align}
\begin{align}\label{eq:lfrak}
\mathfrak{l}_{k+1}
=&  -\frac{\theta}{2}
        a'\left(\Sigma\Sigma'\right)^{-1} a \Delta t
    + \frac{1}{2} 
        \left\{
            - \theta a'\left(\Sigma\Sigma'\right)^{-1} \Sigma
            + b' P_{k+1} \Lambda \Delta t
            + p_{k+1}'\Lambda
            + \theta \Xi'
        \right\} 
        \mathcal{B}_{k+1}^{-1}
                        \nonumber\\   
&   \cdot \left\{
            - \theta \Sigma'\left(\Sigma\Sigma'\right)^{-1} a 
            + \Lambda' P_{k+1} b \Delta t 
            + \Lambda'  p_{k+1} 
            + \theta \Xi
        \right\}
        \Delta t
\end{align}
\end{proposition}

\begin{proof}

We start from the definition of $\gamma^*$ at \eqref{eq:gammastar} and write, for convenience
\begin{align}
    \gamma^*_k &= \mathfrak{K}_{k+1} X_k + \mathfrak{k}_{k+1},
\end{align}
where

\begin{align}
    \mathfrak{K}_{k+1} &= \mathcal{B}_{k+1}^{-1}\left\{
        - \theta \Sigma'\left(\Sigma\Sigma'\right)^{-1} A 
        + \Lambda' P_{k+1} \tilde{B}
        \right\}
                                \nonumber\\
    \mathfrak{k}_{k+1} &= \mathcal{B}_{k+1}^{-1}\left\{
        - \theta \Sigma'\left(\Sigma\Sigma'\right)^{-1} a 
        + \Lambda' P_{k+1} b \Delta t 
        + \Lambda'  p_{k+1} 
        + \theta \Xi
    \right\}.
\end{align}
Next, we use equation \eqref{eq:hstar:1} to express $h^*$ in terms of $\gamma^*$. Recalling that $\eta^* = 0$, equation \eqref{eq:hstar:1} implies that 
\begin{align}
    h^*_k 
    =&  \left(\Sigma\Sigma'\right)^{-1}\left[ 
            (a + A X_k)
       +  \Sigma \gamma^*
       \right] 
                        \nonumber\\
    =&  \left(\Sigma\Sigma'\right)^{-1} \left[
        \left( A + \Sigma \mathfrak{K}_{k+1} \right) X_k
       + \left(a + \Sigma \mathfrak{k}_{k+1}\right)
        \right] 
\end{align}
Then, at the saddle point $(h^*_k, \gamma^*_k, \eta^*_k)$, we express the function $F_k$, which was defined at \eqref{eq:auxfunc:F}, as:
\begin{align}\label{eq:F:statpoint:interm2}
    & F_k(h^*_k, \gamma^*_k, \eta^*_k ; X_k)
                                    \nonumber\\
    =&  X_k' \Bigg\{ 
        \frac{\theta}{2} 
        \left( A + \Sigma \mathfrak{K}_{k+1} \right)' 
        \left(\Sigma\Sigma'\right)^{-1}
        \left( A + \Sigma \mathfrak{K}_{k+1} \right)  \Delta t
        - 
        \theta A' \left(\Sigma\Sigma'\right)^{-1}  
        \left( A + \Sigma \mathfrak{K}_{k+1} \right)
        \Delta t
                                    \nonumber\\ 
    & - \theta\left( A + \Sigma \mathfrak{K}_{k+1} \right)'  
            \left(\Sigma\Sigma'\right)^{-1}\Sigma 
            \mathfrak{K}_{k+1} \Delta t
        + \frac{1}{2} 
            \mathfrak{K}_{k+1}'\left[
                \Lambda' P_{k+1} \Lambda \Delta t - I_d
            \right] 
            \mathfrak{K}_{k+1} \Delta t
                                    \nonumber\\ 
    &   + \mathfrak{K}_{k+1}' \Lambda' P_{k+1} \tilde{B} \Delta t                                   
       \Bigg\}X_k
                                    \nonumber\\
    & + \Bigg\{
        \theta 
        \left(a + \Sigma \mathfrak{k}_{k+1}\right)' 
        \left(\Sigma\Sigma'\right)^{-1}
        \left( A + \Sigma \mathfrak{K}_{k+1} \right)
        \Delta t
        - 
        \theta \left(a + \Sigma \mathfrak{k}_{k+1}\right)' \left(\Sigma\Sigma'\right)^{-1} A \Delta t
                                    \nonumber\\ 
    & - \theta a' \left(\Sigma\Sigma'\right)^{-1}  
        \left( A + \Sigma \mathfrak{K}_{k+1} \right) \Delta t
        - 
        \theta \mathfrak{k}_{k+1}'\Sigma'
            \left(\Sigma\Sigma'\right)^{-1}
            \left( A + \Sigma \mathfrak{K}_{k+1} \right) \Delta t
                                    \nonumber\\ 
    & - \theta\left(a + \Sigma \mathfrak{k}_{k+1}\right)'
            \left(\Sigma\Sigma'\right)^{-1}\Sigma 
            \mathfrak{K}_{k+1} \Delta t
        + \theta \Xi' \mathfrak{K}_{k+1} \Delta t
                                    \nonumber\\
    & + \mathfrak{k}_{k+1}' 
            \left[
                \Lambda' P_{k+1} \Lambda \Delta t - I_d
            \right] 
        \mathfrak{K}_{k+1} \Delta t
        + \left[ 
            P_{k+1} b \Delta t +  p_{k+1} 
        \right]' \Lambda \mathfrak{K}_{k+1} \Delta t        
                                    \nonumber\\
    &   + \mathfrak{k}_{k+1}' \Lambda' P_{k+1}\tilde{B} \Delta t
        \Bigg\}X_k
                                    \nonumber\\
    & + \Bigg\{
        \frac{\theta}{2}
        \left(a + \Sigma \mathfrak{k}_{k+1}\right)' 
        \left(\Sigma\Sigma'\right)^{-1}
        \left(a + \Sigma \mathfrak{k}_{k+1}\right)
        \Delta t
        - 
        \theta a' \left(\Sigma\Sigma'\right)^{-1}  
        \left(a + \Sigma \mathfrak{k}_{k+1}\right)\Delta t
                                    \nonumber\\ 
    & - \theta \left(a + \Sigma \mathfrak{k}_{k+1}\right)'
        \left(\Sigma\Sigma'\right)^{-1}\Sigma 
        \mathfrak{k}_{k+1}\Delta t  
      + \theta \Xi'\mathfrak{k}_{k+1}\Delta t  
                                    \nonumber\\ 
    & + \frac{1}{2} \mathfrak{k}_{k+1}' 
            \left[
                \Lambda' P_{k+1} \Lambda \Delta t - I_d
            \right] 
            \mathfrak{k}_{k+1} \Delta t
        + \mathfrak{k}_{k+1}' \Lambda'
            \left[ 
                P_{k+1} b \Delta t +  p_{k+1} 
            \right] \Delta t
       \Bigg\}
\end{align}
The function $F_k$ is quadratic in $X_k$ at the saddle point $(h^*_k,\gamma^*_k,\eta^*_k)$. We can write it as 
\begin{align}
    F_k(h_k^*, \gamma_k^*, \eta_k^*;X_k)  
    =& \frac{1}{2} X_k' \mathfrak{Q}_{k+1} X_k 
        + X_k' \mathfrak{q}_{k+1} 
        + \mathfrak{l}_{k+1}
\end{align}
where $\mathfrak{Q}_{k+1}$, $\mathfrak{q}_{k+1}$, and $\mathfrak{l}_{k+1}$ are respectively the coefficients of the quadratic term, linear term, and zeroth-order term in \eqref{eq:F:statpoint:interm2}.

Identifying coefficients of the quadratic, linear, and constant terms in \eqref{eq:F:statpoint:interm2}, and then substituting the expressions for $\mathfrak{K}_{k+1}$ and $\mathfrak{k}_{k+1}$, gives \eqref{eq:Qfrak}--\eqref{eq:lfrak}. 

\end{proof}

We can now state our main result. 

\begin{theorem}\label{theo:main:recursions}
Under Assumptions \ref{as:sigma:posdef} and \ref{as:saddlepoint:cond}, the value function $u$ has, as mentioned in \eqref{eq:Phi:quadform0}, a quadratic form of the type
\begin{align}\label{eq:Phi:quadform}
    u_k(X_k) = \frac{1}{2} X_k' P_k X_k + X_k' p_k + r_k,
    \quad k=0, \ldots, K.
\end{align}
where $P_k$, $p_k$, and $r_k$ are deterministic and satisfy the following backward recursions
\begin{enumerate}[(i)]   
\item
\begin{align}\label{eq:recursion:P}
    P_k 
    =&  -\theta A'\left(\Sigma\Sigma'\right)^{-1}A\Delta t
    + \theta^2 
        A'\left(\Sigma\Sigma'\right)^{-1} \Sigma 
        \mathcal{B}_{k+1}^{-1}
        \Sigma'\left(\Sigma\Sigma'\right)^{-1} A 
        \Delta t
                                    \nonumber\\
&   - 2\theta A'\left(\Sigma\Sigma'\right)^{-1}\Sigma   
        \mathcal{B}_{k+1}^{-1}
        \Lambda' P_{k+1} \tilde{B}
        \Delta t   
    + \tilde{B}' P_{k+1} \Lambda 
        \mathcal{B}_{k+1}^{-1}
        \Lambda' P_{k+1} \tilde{B}
        \Delta t 
                                    \nonumber\\
    &   + \tilde{B}' P_{k+1} \tilde{B},
                                    \nonumber\\
    P_T &= 0,
\end{align}

\item
\begin{align}\label{eq:recursion:p}
    p_k 
    =&  - \theta A'\left(\Sigma\Sigma'\right)^{-1} a \Delta t
    + \left\{
        - \theta A'\left(\Sigma\Sigma'\right)^{-1} \Sigma 
        + \tilde{B}' P_{k+1} \Lambda 
    \right\}
    \mathcal{B}_{k+1}^{-1}
                        \nonumber\\
&   \cdot 
        \left\{
        - \theta \Sigma'\left(\Sigma\Sigma'\right)^{-1} a
        + \Lambda' \left[ 
            P_{k+1}b \Delta t +  p_{k+1}  
        \right] 
        + \theta \Xi
    \right\} 
    \Delta t
    + \left[ 
            \tilde{B}' \left(P_{k+1} b  
            + p_{k+1}\frac{1}{\Delta t} \right)
            + \theta C'
        \right] \Delta t,
                                \nonumber\\
        p_T =& 0,            
\end{align}

\item
\begin{align}\label{eq:recursion:r}
    r_k 
    =&  r_{k+1}
    -\frac{\theta}{2}
        a'\left(\Sigma\Sigma'\right)^{-1} a \Delta t
    + \frac{1}{2} 
        \left\{
            - \theta a'\left(\Sigma\Sigma'\right)^{-1} \Sigma
            + b' P_{k+1} \Lambda \Delta t
            + p_{k+1}'\Lambda
            + \theta \Xi'
        \right\} 
        \mathcal{B}_{k+1}^{-1}
                        \nonumber\\   
&   \cdot \left\{
            - \theta \Sigma'\left(\Sigma\Sigma'\right)^{-1} a 
            + \Lambda' P_{k+1} b \Delta t 
            + \Lambda'  p_{k+1} 
            + \theta \Xi
        \right\}
        \Delta t,
                                    \nonumber\\
                                    \nonumber\\
    &  + \left[ 
            \frac{1}{2} b' P_{k+1} b \Delta t
            + b' p_{k+1}   
            + \theta c 
            + \frac{\theta}{2} \tr\left(\Psi_k\Sigma\Sigma'\right) 
            - \frac{\theta}{2}\Xi'\Xi
            +\frac{1}{2} \tr\left(\Lambda'P_{k+1}\Lambda\right) 
        \right]\Delta t
                                    \nonumber\\
    & r_T = 0.              
\end{align}
\end{enumerate}

Furthermore, the candidate optimal controls are given by the saddle point $(h^*,\gamma^*,\eta^*)$ of $F(\bar h,\bar\gamma,\bar\eta)$ according to Proposition \ref{prop:Controls} (formulas at \eqref{eq:hstar}-\eqref{eq:etastar}).

\end{theorem}

\begin{proof} From Propositions \ref{prop:SaddlePoint}  and \ref{P2}, we obtain:
\begin{align}
    u_k(X_k) 
=& \inf_{\bar h } \sup_{\bar \gamma, \bar \eta} \left\{ F_k(\bar h,\bar\gamma,\bar\eta;X_k)
    \right\} 
        + \frac{1}{2}\left(b \Delta t + \tilde{B} X_k\right)' P_{k+1} \left(b \Delta t + \tilde{B} X_k\right)
        + \left(b \Delta t + \tilde{B} X_k \right)'p_{k+1}    
                                    \nonumber\\
    &  + \theta (c + C X_k) \Delta t   
        + \frac{\theta}{2} \tr\left( \Psi_k\Sigma\Sigma'\right)\Delta t 
        - \frac{\theta}{2}\Xi'\Xi \Delta t
        +\frac{1}{2} \tr\left(\Lambda' P_{k+1}\Lambda\right)\Delta t
        + r_{k+1}
                                    \nonumber\\
=&  \frac{1}{2} X_k' \mathfrak{Q}_{k+1} X_k 
        + X_k' \mathfrak{q}_{k+1}  
        + \mathfrak{l}_{k+1}
        + \frac{1}{2}  X_k' \tilde{B}' P_{k+1} \tilde{B} X_k
        + \left[ 
            b' P_{k+1} \tilde{B} 
            + p_{k+1}' \tilde{B}\frac{1}{\Delta t} 
            + \theta C
        \right]X_k \Delta t
                                    \nonumber\\
    &  + \left[ 
            \frac{1}{2} b' P_{k+1} b \Delta t
            + b' p_{k+1}   
            + \theta c 
            + \frac{\theta}{2} \tr\left(\Psi_k\Sigma\Sigma'\right) 
            - \frac{\theta}{2}\Xi'\Xi
            +\frac{1}{2} \tr\left(\Lambda'P_{k+1}\Lambda\right) 
        \right]\Delta t
        + r_{k+1}
\end{align}
with $\mathfrak{Q}_{k+1}$, $\mathfrak{q}_{k+1}$, $\mathfrak{l}_{k+1}$ given at \eqref{eq:Qfrak}--\eqref{eq:lfrak}. Equate this expression with the quadratic form at \eqref{eq:Phi:quadform}. The recursions for $P_k$, $p_k$, and $r_k$ at \eqref{eq:recursion:P}-\eqref{eq:recursion:r} follow.

Moreover, the candidate controls given by the saddle point $(h^*,\gamma^*,\eta^*)$ of $F(\bar h,\bar\gamma,\bar\eta)$ given at \eqref{eq:hstar}-\eqref{eq:etastar} are admissible, in the sense that $h^* \in \bar{\mathcal A}^H_{\text{expl}}$, $\gamma^* \in \bar{\mathcal A}^{\bar \Gamma}$, and $\eta^* \in \bar{\mathcal A}^{\bar \eta}$, so they are optimal. Consequently, 
\begin{align}
    & \sup_{\bar\gamma, \bar\eta} \mathbf{E}^{\bar\gamma, \bar\eta} \left[ -\theta (\bar R^\pi_T-R_0)  
    		-\frac{1}{2} \sum_{k=0}^{K-1} \left( \|\bar\gamma_k\|^2 \Delta t + \bar\eta_k' \Psi_k^{-1} \bar\eta_k \right)
   	 \right]
                                \nonumber\\
    &=\sup_{\bar \gamma \in \mathcal{A}^{\bar \Gamma}, \bar\eta \in \mathcal{A}^{\bar \eta}} \mathbf{E}^{\bar\gamma, \bar\eta} \left[ -\theta (\bar R^\pi_T-R_0)  
    		-\frac{1}{2} \sum_{k=0}^{K-1} \left( \|\bar\gamma_k\|^2 \Delta t + \bar\eta_k' \Psi_k^{-1} \bar\eta_k \right)
   	 \right],
\end{align}
which shows that Assumption \ref{as:Gamma:Eta} holds. 
\end{proof}

\begin{corollary}
    We have
    \begin{align}\label{eq:optI}
        \inf_{\bar h \in \bar{\mathcal A}^{H}_{\mathrm{expl}}} I(H,\theta) = \exp\left\{\frac{1}{2} X_0' P_0 X_0 + X_0' p_0 + r_0 \right\}
    \end{align}
    and 
    \begin{align}\label{eq:optJ}
        \sup_{\bar h \in \bar{\mathcal A}^{H}_{\mathrm{expl}}} J(H,\theta) = -\frac{1}{\theta}\left( \frac{1}{2} X_0' P_0 X_0 + X_0' p_0 + r_0\right)
    \end{align}
    with $P_k$, $p_k$, and $r_k$ given by \eqref{eq:recursion:P},\eqref{eq:recursion:p}, and \eqref{eq:recursion:r}, and optimal controls given by the saddle point $(h^*,\gamma^*,\eta^*)$ according to Proposition \ref{prop:Controls}.
\end{corollary}

\begin{proof}
Start from the criterion $I$:
\begin{align}
    \inf_{\bar h \in \bar{\mathcal A}^{H}_{\mathrm{expl}}} I(H,\theta)
    &\underset{\eqref{eq:EEDuality:inf}}{=} \exp\left\{ \inf_{\bar h \in \bar{\mathcal A}^{H}_{\mathrm{expl}}} \sup_{\bar\gamma \in \mathcal{A}^{\bar \Gamma},\bar\eta \in \mathcal{A}^{\bar \eta}} \mathbf{E}^{\bar\gamma, \bar\eta} \Bigg[ \theta \sum_{k=0}^{K-1} g(X_k, \bar h_k, \bar\eta_k, \bar \gamma_k)\Delta t \Bigg] \right\}
                        \nonumber\\
    &\underset{\eqref{eq:criterion:I:Pbar}}{=}
    \exp \left\{ u_0(X_0) \right\}
                            \nonumber\\
    &\underset{\eqref{eq:Phi:quadform} \, \text{at} \, k=0}{=}
    \exp \left\{ \frac{1}{2} X_0' P_0 X_0 + X_0' p_0 + r_0 \right\},
\end{align} 
which gives \eqref{eq:optI}. Furthermore, 
\begin{align}
    \ln \inf_{\bar h \in \bar{\mathcal A}^{H}_{\mathrm{expl}}} I(H,\theta)
    \underset{\eqref{eq:EEDuality:inf}}{=} 
    \inf_{\bar h \in \bar{\mathcal A}^{H}_{\mathrm{expl}}} \ln I(H,\theta)
    \underset{\eqref{eq:I}}{=} \inf_{\bar h \in \bar{\mathcal A}^{H}_{\mathrm{expl}}} -\theta J(H,\theta)
    = -\theta \sup_{\bar h \in \bar{\mathcal A}^{H}_{\mathrm{expl}}}  J(H,\theta)
                            \nonumber\\
\end{align}
from which \eqref{eq:optJ} follows.
\end{proof}

\subsection{Limiting Case: The Kelly Criterion}\label{sec:KellyandFKS:Kelly}

The Kelly portfolio allocation problem involves the maximization of the expected log utility. It leads to a stochastic control problem that can be seen either as the limiting case of the Linear-Quadratic-Gaussian problem studied above when $\theta \to 0$ or as an LQG control problem with the criterion
\begin{align}\label{eq:J:Kelly} 
J^\text{Kelly}(H)
&:= \mathbf{E} \left[ R_T - R_0\right]
                        \nonumber\\
&= \mathbf{E} \left[ \sum_{k=0}^{K-1} \left\{ 
    \left(- \frac{1}{2} h_k'\Sigma\Sigma'h_k 
    + h_k'a 
    + \frac{1}{2}\Xi'\Xi- c \right)  
    + \left(h_k' A - C \right)X_k\right\} \Delta t
    \right.
                                \nonumber\\
    &\left.
    + \sum_{k=0}^{K-1} \left(h_k'\Sigma - \Xi' \right) w_k
  \right]
                          \nonumber\\
&= \mathbf{E} \left[ \sum_{k=0}^{K-1} \left\{ 
    \left(- \frac{1}{2} h_k'\Sigma\Sigma'h_k 
    + h_k'a 
    + \frac{1}{2}\Xi'\Xi- c \right)  
    + \left(h_k' A - C \right)X_k\right\} \Delta t
    \right].
\end{align}

The instantaneous reward function has a quadratic structure, so the control problem is LQG. We can solve it directly under the measure $\mathbb{P}$. 

We define the optimal value function $u^\text{Kelly}$ for the Kelly portfolio allocation problem as
\begin{align}\label{eq:valuefunction:u:Kelly}
    u^\text{Kelly}_0(X_0) 
    := \sup_{H \in \mathcal{A}^H} \mathbf{E} \left[ R_T-R_0\right] 
    = \sup_{H \in \mathcal{A}^H} \mathbf{E} \left[ \sum_{k=0}^{K-1} g^\text{Kelly}\left(X_k, h_k\right) \Delta t
  \right].
\end{align}
where
\begin{align}\label{eq:g:Kelly}
    g^\text{Kelly}\left(X_k, h_k\right) 
    :=  \left(- \frac{1}{2} h_k'\Sigma\Sigma'h_k 
    + h_k'a 
    + \frac{1}{2}\Xi'\Xi- c \right)  
    + \left(h_k' A - C \right)X_k.
\end{align}
The state process $X$ evolves according to \eqref{eq:state}, that is,  
\begin{align}
    X_{k+1} = b \Delta t + \tilde{B} X_k + \Lambda w_k.
\end{align}
Importantly, the state is uncontrolled.

We apply the Dynamic Programming Principle (DPP) to express the value function recursively as:
\begin{align}\label{eq:DPP:TC:Kelly}
u^\text{Kelly}_T(X_T)=& 0,
\end{align}
and, for $k=K-1,\cdots,0$, 
\begin{align}\label{eq:DPP:DPP:Kelly} 
u^\text{Kelly}_k(X_k)
    =&\sup_{h_k} \mathbf{E}_{k,X_k} \left[
        g^\text{Kelly}\left(X_k, h_k\right) \Delta t
        + u^\text{Kelly}_{k+1}(X_{k+1}) \right]
                                  \nonumber\\ 
    =&\sup_{h_k}  \left\{
        g^\text{Kelly}\left(X_k, h_k\right)\Delta t
        + \mathbf{E}_{k,X_k} \left[ u^\text{Kelly}_{k+1}(X_{k+1}) \right]\right\}.
\end{align}
We can show that $u^\text{Kelly}_k(X_k)$ has a quadratic expression in $X_k$ of the form 
\begin{equation}\label{eq:Phi:quadform0:Kelly}
    u^\text{Kelly}_k(X_k) = \frac{1}{2}X_k' P^\text{Kelly}_k X_k + X_k' p^\text{Kelly}_k + r^\text{Kelly}_k,
\end{equation}
where $P^\text{Kelly}_k$, $p^\text{Kelly}_k$, and $r^\text{Kelly}_k$ can be computed recursively.

Since $X_{k+1}$ does not depend on $h_k$, the continuation value $u^\text{Kelly}_{k+1}(X_{k+1})$ does not affect the first-order condition. Therefore, we can get the optimal control directly from a pointwise maximization of the function $g^\text{Kelly}$. We obtain
\begin{align}\label{eq:h:Kelly}
    h^\text{Kelly}_k = \left(\Sigma\Sigma'\right)^{-1}(a + A X_k).
\end{align}

\section{Interpreting the Risk-Sensitive Investment Management Model}\label{sec:interpretation}

\subsection{Interpreting Assumptions \ref{as:sigma:posdef} and \ref{as:saddlepoint:cond}}\label{sec:OptimalityConditions}

Assumptions \ref{as:sigma:posdef} and \ref{as:saddlepoint:cond} are necessary for Propositions \ref{prop:SaddlePoint} and \ref{prop:Controls} to hold. 

Assumption \ref{as:sigma:posdef} is a standard condition in portfolio optimization. It prevents the investment universe from containing assets with identical risk profiles, such as a commodity and a futures contract on that commodity (e.g., oil and oil futures), or the S\&P 500 alongside its value and growth subindices.

Assumption \ref{as:saddlepoint:cond} requires the block matrix to be negative definite, which is equivalent to satisfying both  $\mathcal{A}_{k+1}  = \Lambda' P_{k+1} \Lambda \Delta t - I_d <0$ and $\theta \Sigma\Sigma' \Delta t - \Psi_k^{-1} < 0$. The first subcondition, $\Lambda' P_{k+1} \Lambda \Delta t - I_d <0$, is equivalent to $\Lambda' P_{k+1} \Lambda \Delta t < I_d $. The term $\Lambda' P_{k+1} \Lambda$ represents the curvature of the value function projected onto the state-noise directions. Hence, the first subcondition requires that the projected curvature of the value function, when restricted to the noise subspace and scaled by $\Delta t$, has all eigenvalues strictly smaller than one. This is a regularity constraint on the propagation of the state noise through the curvature of the value function.

The second subcondition, $\theta \Sigma\Sigma' \Delta t - \Psi_k^{-1} < 0$, imposes the following constraint on the inverse of the covariance matrix of exploration $\Psi_k$: 
\begin{align}\label{cond:explo:covar}
    \Psi_k^{-1} > \theta \Sigma\Sigma' \Delta t.
\end{align}
So $\Psi_k$ is inversely proportional to the investor's risk sensitivity, $\theta$, the diffusion matrix of the asset returns, $\Sigma\Sigma'$, and the time step, $\Delta t$. These relations admit a clear economic interpretation. 

Exploration increases the variability of portfolio weights, so more risk-sensitive investors will explore less. As their risk sensitivity declines, investors have more freedom to explore. In the limit as $\theta \to 0$, the exploration constraint becomes independent of risk sensitivity, so Kelly investors do not face a risk-sensitivity-induced upper bound on exploration.

Asset return volatility already induces variability in portfolio weights and therefore generates indirect exploration of the state variable. Therefore, investors should explore less if their investment universe contains highly volatile securities. 

When the time interval between rebalancings is short, investors can afford greater exploration because they can adjust their asset allocation more frequently. With a longer time interval between rebalancings, departures from the optimal asset allocation become costlier.

\subsection{Optimality Conditions Implied by Proposition \ref{prop:Controls:Alt}}\label{sec:OptimalityConditions:Alt}

Equation \eqref{eq:SOC:hstar:Alt} in the proof of Proposition \ref{prop:Controls:Alt} implies that at the saddle point $(h^*_k,\gamma^*_k,\eta^*_k)$,
\begin{align}
    \mathcal{C}_{k+1}\!(\theta) > 0  
\end{align}
Recalling the definition of $\mathcal{C}_{k+1}$ at \eqref{eq:CalC}, we rewrite this inequality using the initial data for the control problem:
\begin{align}
    \frac{1}{\theta+1}\Sigma\left[
        I_d
        - \theta \left(\Lambda' P_{k+1} \Lambda \Delta t - I_d\right)^{-1} 
        \right]\Sigma' > 0. 
\end{align}
This condition is a generalized \emph{risk-resistance condition}, a condition typically required to solve risk-sensitive control problems \citep{ShaijuPetersenFormulasLQR_LQ_LEQG2008,Whittle1990}. Specifically, the term $I_d - \theta \left(\Lambda' P_{k+1} \Lambda \Delta t - I_d\right)^{-1}$ plays the role of a risk-resistance factor in noise space. The pre- and post-multiplication by 
$\Sigma$ and $\Sigma'$ map this noise-space condition into portfolio space. This operation is intuitive. In our setup, investors control the state variable (the risk factors) through their holdings of risky assets. Moreover, the state variable noise and the noise associated with the risky asset are correlated. 

Additionally, equation \eqref{eq:SOC:etastar:Alt} in the proof of Proposition \ref{prop:Controls:Alt} gives an upper bound for the covariance of exploration in terms of the risk-resistance condition:
\begin{align}
    \Psi_k^{-1} >
        \theta \Sigma
        \left\{
            I_d
            - \theta \left(\Lambda' P_{k+1} \Lambda \Delta t - I_d\right)^{-1} 
        \right\}\Sigma'\Delta t.
\end{align}

To conclude, both sets of conditions ultimately compare the curvature of the value function with the intrinsic geometry of Gaussian noise. Moreover, an important advantage of applying the Free Energy-Entropy Duality is that it only requires Assumptions \ref{as:sigma:posdef} and \ref{as:saddlepoint:cond}. These assumptions are both weaker and more directly verifiable than the classical risk-resistance condition, which typically requires recursive spectral bounds at each step.

\subsection{Optimal Investment Strategies as Kelly Strategies}\label{sec:KellyandFKS}

We now turn our attention to interpreting the optimal asset allocation $h^*$. We start by showing that $h^*$ can be represented as an allocation to the Kelly portfolio $h^\text{Kelly}$ from \eqref{eq:h:Kelly}, penalized by the control $\gamma^*$ that is induced by the free energy-entropy duality penalization and the choice of an optimal measure $\mathbb{P}^{\gamma^*,\eta^*}$:

\begin{corollary}[Corollary to Proposition \ref{prop:SaddlePoint} - Kelly Strategy, Part I: Penalized Kelly Allocation]\label{coro:penalizedKelly}

The optimal investment strategy is to invest in both the Kelly portfolio and an allocation
$\left(\Sigma\Sigma'\right)^{-1}\Sigma \gamma^*$ related to the penalizing control $\gamma^*$ defined at \eqref{eq:gammastar}:  
\begin{align}\label{eq:KellyStrategy1}
h^*_k =  h^\text{Kelly}_k
       +  \left(\Sigma\Sigma'\right)^{-1}\Sigma \gamma^*_k,
\end{align}
    
\end{corollary}

\begin{proof}
Equation \eqref{eq:hstar:1} in the proof of Proposition \ref{prop:SaddlePoint} gives the following expression for the mean $(h^* + \eta^*)$ of the distribution of exploratory policies under $\mathbb{P}^{\bar\gamma,\bar\eta}$:
\begin{align*}
    \left(h^* + \eta^*\right)
        =  \left(\Sigma\Sigma'\right)^{-1}\left[ 
            (a + A X_k)
       +  \Sigma \gamma^*_k
       \right]
\end{align*}
Since $\eta^*=0$, and applying the definition of the Kelly portfolio at \eqref{eq:h:Kelly}, we have 
\begin{align}
h^*_k 
=  \left(\Sigma\Sigma'\right)^{-1}(a + A X_k)
        +  \left(\Sigma\Sigma'\right)^{-1}\Sigma \gamma^*_k
 = h^\text{Kelly}_k
        +  \left(\Sigma\Sigma'\right)^{-1}\Sigma \gamma^*_k,
\end{align}
which proves the claim.
\end{proof} 

Intuitively, the structure of equation \eqref{eq:KellyStrategy1} mirrors that of the free energy-entropy duality. The duality expresses a risk-sensitive control problem as a risk-neutral LQG problem, penalized by the antagonistic control $\bar \gamma$. The optimal investment strategy at \eqref{eq:KellyStrategy1} is the Kelly strategy, which is the optimal control for the risk-neutral LQG problem, penalized by a term related to the optimal antagonistic control $\gamma^*$. This penalty term, $\left(\Sigma\Sigma'\right)^{-1}\Sigma \gamma^*$, also has a clean geometric interpretation. It is the image of the $\mathbb{R}^d$-dimensional optimal antagonistic control $\gamma^*$ through the linear map $(\Sigma\Sigma')^{-1}\Sigma$, providing an intuitive connection with least-squares-type linear mappings.

The next result is inspired by the continuous-time risk-sensitive asset management literature and motivates an interpretation of the optimal asset allocation as a fractional Kelly strategy. For later comparison, we first recall the standard definition of fractional Kelly strategies obtained in continuous-time problems \citep[see in particular][]{dall_RSBench,DavisLleoBook2014,LleoRunggaldier_EntropyRegularizationinRLandRSIM_2026}:

\begin{proposition}[Fractional Kelly Strategy (FKS) for Risk-Sensitive Benchmarked Investment Management in Continuous Time - adapted from Proposition 2.10 in \citep{LleoRunggaldier_EntropyRegularizationinRLandRSIM_2026}]\label{prop:FKS}
\;

The continuous-time optimal benchmarked investment strategy $h^*(s,X_s), s \in [t,T]$ consists of an allocation between three funds: $h^\text{Kelly}$, $h^\text{Bench}$, and $h^{IHP}$.

\begin{enumerate}[(i)]

    \item The fund $h^\text{Kelly}$ is a Kelly portfolio with factor-dependent allocation
    \begin{align}\label{eq:Kelly:cont}
        h^\text{Kelly}(X_s) := \left(\Sigma\Sigma'\right)^{-1}\left(a + A X_s \right).
    \end{align}

    \item The fund $h^\text{Bench}$ is a benchmark-tracking portfolio with constant allocation
    \begin{align}\label{eq:BenchTracking:cont}
            h^\text{Bench} := \left(\Sigma\Sigma'\right)^{-1}\Sigma\Xi'.
    \end{align}

    \item The fund $h^{IHP}$ is an Intertemporal Hedging Portfolio (IHP) with factor-dependent allocation
    \begin{align}\label{eq:IHP:cont}
        h^{IHP}(s, X_s) := \left(\Sigma\Sigma'\right)^{-1}\Sigma\Lambda'Du^\text{c}(s,X_s),
    \end{align}
    where $u^\text{c}$ is the optimal value function for the continuous-time risk-sensitive benchmarked investment management problem and where $(Du^\text{c})'(s,x) = \left( \frac{\partial u^\text{c}}{\partial x_1}(s,x), \ldots, \frac{\partial u^\text{c}}{\partial x_n}(s,x)\right)$.
\end{enumerate}

Moreover, the relative allocation of the funds is constant at $f := \frac{1}{\theta+1}$ for $h^\text{Kelly}$, $1-f$ for $h^\text{Bench}$, and $f-1$ for $h^{IHP}$.

\end{proposition}

In our setting with continuously-priced stocks and discretely-estimated factors, the optimal asset allocation \eqref{eq:hstar:alt} in Proposition \ref{prop:Controls:Alt} can be rewritten as a \emph{rotated and rescaled} fractional Kelly strategy, as shown in the next proposition. 

\begin{proposition}[Rotated and Rescaled Fractional Kelly Strategy]\label{prop:RKS:1}

The optimal investment strategy $h^*$ at \eqref{eq:hstar:alt} consists of a rotated and rescaled allocation to the Kelly portfolio $h^\text{Kelly}$ defined at \eqref{eq:Kelly:cont} but here evaluated at discrete time $k$, benchmark-tracking portfolio $h^\text{Bench}$ defined at \eqref{eq:BenchTracking:cont}, and a discrete-time intertemporal hedging portfolio $h^\text{IHP}$, defined as
\begin{align}\label{eq:IHP:disc}
    h_k^{\text{IHP}}
    :=& \left(\Sigma\Sigma'\right)^{-1}\Sigma\Lambda' 
        \left[ 
            P_{k+1} \left(b \Delta t + \tilde{B} X_k \right) 
            +  p_{k+1} 
        \right]
\end{align}
Moreover, the rotation and rescaling of these portfolios is performed within the optimal investment strategy $h^*$ as:
\begin{align}
    h^*_k = 
    \frac{1}{\theta+1} \mathcal{C}_{k+1}^{-1}\!(\theta)\Sigma\Sigma' h^\text{Kelly}_k 
    + \frac{\theta}{\theta+1} \mathcal{C}_{k+1}^{-1}\!(\theta)\mathcal{A}_{k+1}^{-1}\Sigma\Sigma' h^\text{Bench} 
    - \frac{\theta}{\theta+1} \mathcal{C}_{k+1}^{-1}\!(\theta)\mathcal{A}_{k+1}^{-1}\Sigma\Sigma' h^\text{IHP}_k,
\end{align}
where $\mathcal{C}_{k+1}^{-1}\!(\theta)\Sigma\Sigma'$ is the rescaling factor for the Kelly portfolio and $\frac{\theta}{\theta+1} \mathcal{C}_{k+1}^{-1}\!(\theta)\mathcal{A}_{k+1}^{-1}\Sigma\Sigma'$ is the rescaling factor for the benchmark-tracking and intertemporal hedging portfolioss. 

\end{proposition}

\begin{proof}
We recall the optimal asset allocation given at \eqref{eq:hstar:alt}:
\begin{align}
    h^*_k
    =&  \frac{1}{\theta+1} \mathcal{C}_{k+1}^{-1}\!(\theta)
    \Bigg[
        (a + A X_k)
        + \theta \mathcal{A}_{k+1}^{-1} \Sigma
        \left\{
            \Xi
            - \Lambda'\left[ 
                P_{k+1} \left(b \Delta t + \tilde{B} X_k \right) 
                +  p_{k+1} 
            \right]
        \right\}
    \Bigg],
\end{align}

Define
\begin{align}\label{eq:Kelly:rotatedrescaled:v1}
    h^{\widetilde{\text{Kelly}}}_k :=
        \mathcal{C}_{k+1}^{-1}\!(\theta)(a + A X_k),
\end{align}
\begin{align}\label{eq:BenchTracking:rotatedrescaled:v1}
    h^{\widetilde{\text{Bench}}}_k := 
            \mathcal{C}_{k+1}^{-1}\!(\theta)\mathcal{A}_{k+1}^{-1} \Sigma\Xi.
\end{align}  
and  
\begin{align}\label{eq:IHP:rotatedrescaled:v3}
 h^{\widetilde{\text{IHP}}}_k 
        := -\mathcal{C}_{k+1}^{-1}\!(\theta)\mathcal{A}_{k+1}^{-1}\Sigma\Lambda' 
        \left[ 
            P_{k+1} \left(b \Delta t + \tilde{B} X_k \right) 
            +  p_{k+1} 
        \right]
\end{align}
With these definitions, we express the optimal asset allocation at \eqref{eq:hstar:alt} as:
\begin{align}\label{eq:FKS:Rescaled:Interm}
    h^*_k = 
    \frac{1}{\theta+1} h^{\widetilde{\text{Kelly}}}_k 
    + \frac{\theta}{\theta+1} h^{\widetilde{\text{Bench}}}_k 
    + \frac{\theta}{\theta+1} h^{\widetilde{\text{IHP}}}_k.
\end{align}

Next, we relate $h^{\widetilde{\text{Kelly}}}$, and $h^{\widetilde{\text{Bench}}}$, and $h^{\widetilde{\text{IHP}}}$ with the continuous-time Kelly portfolio $h^\text{Kelly}$ defined at \eqref{eq:Kelly:cont}, the continuous-time benchmark-tracking portfolio $h^\text{Bench}$ defined at \eqref{eq:BenchTracking:cont}, and the discrete-time intertemporal hedging portfolio $h^\text{IHP}$, defined at \eqref{eq:IHP:disc}.

We start with:
\begin{align}
    h^{\widetilde{\text{Kelly}}}_k 
    = \mathcal{C}_{k+1}^{-1}\!(\theta)(a + A X_k)
    = \mathcal{C}_{k+1}^{-1}\!(\theta)(\Sigma\Sigma')(\Sigma\Sigma')^{-1}(a + A X_k)
    = \mathcal{C}_{k+1}^{-1}\!(\theta)\Sigma\Sigma' h^\text{Kelly}_k.
\end{align}
Similarly,
\begin{align}\label{eq:BenchTracking:rotatedrescaled:v2}
    h^{\widetilde{\text{Bench}}}_k 
    = \mathcal{C}_{k+1}^{-1}\!(\theta)\mathcal{A}_{k+1}^{-1} \Sigma\Xi
    = \mathcal{C}_{k+1}^{-1}\!(\theta)\mathcal{A}_{k+1}^{-1}\Sigma\Sigma'h^\text{Bench}
\end{align}  
and
\begin{align}\label{eq:IHP:rotatedrescaled:v1}
 h^{\widetilde{\text{IHP}}}_k 
    = -\mathcal{C}_{k+1}^{-1}\!(\theta)\mathcal{A}_{k+1}^{-1}\Sigma\Lambda' 
        \left[ 
            P_{k+1} \left(b \Delta t + \tilde{B} X_k \right) 
            +  p_{k+1} 
        \right]
    = -\mathcal{C}_{k+1}^{-1}\!(\theta)\mathcal{A}_{k+1}^{-1}\Sigma\Sigma'h^\text{IHP}_k.
\end{align}
Substituting into \eqref{eq:FKS:Rescaled:Interm}, we get
\begin{align}
    h^*_k = 
    \frac{1}{\theta+1} \mathcal{C}_{k+1}^{-1}\!(\theta)\Sigma\Sigma' h^\text{Kelly}_k 
    + \frac{\theta}{\theta+1} \mathcal{C}_{k+1}^{-1}\!(\theta)\mathcal{A}_{k+1}^{-1}\Sigma\Sigma' h^\text{Bench} 
    - \frac{\theta}{\theta+1} \mathcal{C}_{k+1}^{-1}\!(\theta)\mathcal{A}_{k+1}^{-1}\Sigma\Sigma' h^\text{IHP}_k,
\end{align}
which proves the claim.
\end{proof}

\begin{remark}
Note that $h^\text{Bench}$ is constant because $\Sigma$ and $\Xi$ are both assumed constant. By contrast, $h^{\widetilde{\text{Bench}}}_k$ depends on the time index $k$ via the matrices $\mathcal{C}_{k+1}^{-1}\!(\theta)$ and $\mathcal{A}_{k+1}^{-1}$.    
\end{remark}

Therefore, the optimal asset allocation in our model can be interpreted as a fractional Kelly strategy, but with the relative asset allocations of the constituent portfolios---Kelly, benchmark-tracking, and intertemporal hedging---rotated and rescaled by the inverse of the risk-resistance matrix $\mathcal{C}_{k+1}^{-1}\!(\theta)$, discussed above in Section \ref{sec:OptimalityConditions}. Geometrically, this inverse transformation shrinks directions associated with higher resistance (larger eigenvalues of $\mathcal{C}_{k+1}\!(\theta)$) and amplifies directions associated with lower resistance. 

Specifically, the rotation and rescaling factor for the Kelly portfolio is $\frac{1}{\theta+1} \mathcal{C}_{k+1}^{-1}\!(\theta)\Sigma\Sigma'$, while that for the other two constituent portfolios is $\frac{\theta}{\theta+1} \mathcal{C}_{k+1}^{-1}\!(\theta)\mathcal{A}_{k+1}^{-1}\Sigma\Sigma'$. So the rotation and rescaling for the benchmark-tracking and intertemporal hedging portfolios depend also on the inverse of the matrix $\mathcal{A}_{k+1}$, discussed above in Section \ref{sec:OptimalityConditions}. 

For completeness, we also propose a Kelly-like strategy based on the asset allocation in Proposition \ref{prop:Controls}:

\begin{proposition}[Rotated and Rescaled Fractional Kelly Strategy - Part II]\label{prop:RKS:2}
\;

The optimal investment strategy $h^*$ at \eqref{eq:hstar} admits a decomposition into three components: a rotated and rescaled allocation to the Kelly portfolio $h^\text{Kelly}$, a benchmark-tracking component, and an intertemporal hedging component. It can be written as
\begin{align}
    h^*_k = 
    \frac{1}{\theta+1} h^{\widetilde{\text{Kelly}}}_k 
    + \frac{\theta}{\theta+1} h^{\widetilde{\text{Bench}}}_k
    + \frac{\theta}{\theta+1} h^{\widetilde{\text{IHP}}}_k.
\end{align}
where
\begin{align}\label{eq:Kelly:rotatedrescaled:RKS:2}
    h^{\widetilde{\text{Kelly}}}_k :=
    (\theta+1)\left\{
        I_m
        - \theta \left(\Sigma\Sigma'\right)^{-1}\Sigma \mathcal{B}_{k+1}^{-1} \Sigma'
    \right\} h^\text{Kelly}_k,
\end{align}
\begin{align}\label{eq:BenchTracking:rotatedrescaled:RKS:2}
        h^{\widetilde{\text{Bench}}}_k := (\theta+1)\left(\Sigma\Sigma'\right)^{-1}\Sigma \mathcal{B}_{k+1}^{-1} \Xi,
\end{align}
and
\begin{align}\label{eq:IHP:rotatedrescaled:RKS:2}
    h^{\widetilde{\text{IHP}}}_k 
    := \frac{\theta+1}{\theta}\left(\Sigma\Sigma'\right)^{-1} \Sigma \mathcal{B}_{k+1}^{-1}\Lambda' \left\{  
        P_{k+1} \left(b \Delta t + \tilde{B} X_k \right)
        + p_{k+1}
    \right\}
\end{align}
\end{proposition}

\begin{proof}
The decomposition follows directly from equations \eqref{eq:hstar}, that is,
\begin{align}
    h^*_k
    =&  \left(\Sigma\Sigma'\right)^{-1}\Bigg\{ 
        \left\{
            I_m
            - \theta \Sigma \mathcal{B}_{k+1}^{-1} \Sigma'\left(\Sigma\Sigma'\right)^{-1}
        \right\}(a + A X_k)
        + \Sigma \mathcal{B}_{k+1}^{-1}\left\{  
            \Lambda' P_{k+1} \left(b \Delta t + \tilde{B} X_k \right)
            + \Lambda'   p_{k+1} 
            + \theta \Xi
        \right\}
       \Bigg\}.
\end{align}

Define:
\begin{align}\label{eq:Kelly:rotatedrescaled:v2}
    h^{\widetilde{\text{Kelly}}}_k :=
    (\theta+1)\left(\Sigma\Sigma'\right)^{-1} 
    \left\{
        I_m
        - \theta \Sigma \mathcal{B}_{k+1}^{-1} \Sigma'\left(\Sigma\Sigma'\right)^{-1}
    \right\}(a + A X_k),
\end{align}
$h^{\widetilde{\text{Bench}}}_k$ according to \eqref{eq:BenchTracking:rotatedrescaled:RKS:2}, and $h^{\widetilde{\text{IHP}}}_k$ according to \eqref{eq:IHP:rotatedrescaled:RKS:2}. With these definitions, we can express the optimal asset allocation as:
\begin{align}
    h^*_k = 
    \frac{1}{\theta+1} h^{\widetilde{\text{Kelly}}}_k 
    + \frac{\theta}{\theta+1} h^{\widetilde{\text{Bench}}}_k
    + \frac{\theta}{\theta+1} h^{\widetilde{\text{IHP}}}_k.
\end{align}

We proceed as in the proof of Proposition \ref{prop:RKS:1} and relate $h^{\widetilde{\text{Kelly}}}$ with the Kelly portfolio $h^\text{Kelly}$.
\begin{align}
    h^{\widetilde{\text{Kelly}}}_k 
    &= (\theta+1)\left(\Sigma\Sigma'\right)^{-1} 
    \left\{
        I_m
        - \theta \Sigma \mathcal{B}_{k+1}^{-1} \Sigma'\left(\Sigma\Sigma'\right)^{-1}
    \right\}(a + A X_k)
                                \nonumber\\
    &= (\theta+1)\left(\Sigma\Sigma'\right)^{-1} 
    \left\{
        I_m
        - \theta \Sigma \mathcal{B}_{k+1}^{-1} \Sigma'\left(\Sigma\Sigma'\right)^{-1}
    \right\} \left(\Sigma\Sigma'\right) \left(\Sigma\Sigma'\right)^{-1}(a + A X_k)
                                \nonumber\\
    &= (\theta+1)\left\{
        I_m
        - \theta \left(\Sigma\Sigma'\right)^{-1}\Sigma \mathcal{B}_{k+1}^{-1} \Sigma'
    \right\} h^\text{Kelly}_k.
\end{align}
This proves the decomposition and the representation \eqref{eq:Kelly:rotatedrescaled:RKS:2} of the rotated and rescaled Kelly component.
\end{proof}



\section{Numerical Solution via Reinforcement Learning}\label{sec:RL}

So far, we have implicitly assumed that the market model introduced in Section \ref{sec:setup} is an exact representation of financial markets and that all its parameters can be estimated with high accuracy. When these assumptions hold, we can use the recursions in Theorem \ref{theo:main:recursions} to numerically solve the risk-sensitive benchmarked investment management problem. However, in reality, estimating the parameters with sufficient accuracy is extremely difficult. The parameters determining expected returns are a typical example. Accordingly, the main purpose of the market model constructed in Sections \ref{sec:setup} and \ref{sec:solution} is to capture the essential features of financial markets, to characterize their fundamental relations, and to provide a functional form for the policies and value functions. 

We can then use these insights to engineer a reinforcement learning (RL) method that will learn the investment strategy and, possibly, the value function, directly through interactions with the market. One major class of RL techniques is based on policy gradient methods, which start from parametrized families of controls and value functions, where the parameters have to be learned through interaction with the market. On the other hand, for the given model structure in Section \ref{sec:setup}, and assuming that the model parameters are known, we derived explicit expressions for the value functions and the control policies in Section \ref{sec:solution}. This gives us an indication of how to choose parametrized expressions for the value functions and controls with parameters that have to be learned from observing the market. To learn these parameters, we shall apply policy gradient methods. These parametrized expressions are actually approximations to the actual expressions for the value function and controls, and a great deal of the literature is devoted to studying their convergence. In the present section, we discuss how to implement policy gradient methods to solve our investment problem numerically when only the model structure, but not the parameters, is assumed to be known.

Policy gradient methods have gained popularity in recent years following groundbreaking advances in the field of reinforcement learning and a series of successful applications, most notably to Go. We refer the reader to Chapter 13 in the classic text by \citet{SuttonBarto2018}. All policy gradient methods follow the same essential steps. First, approximate the policy using a conveniently parametrized function. Second, compute the gradient of a given performance measure, usually a value function, with respect to the parameters of the function approximating the controls. Third, update the parameters by performing a gradient descent update if the aim is to minimize the performance measure, or a gradient ascent update if the aim is to maximize the performance measure.

In this section, we use the typewriter font for the learnable parameter matrices and vectors
$\mathtt{D},\mathtt{d},\mathtt{E},\mathtt{e},\mathtt{F},\mathtt{f}$, and we use $\phi_k$ and $\boldsymbol{\Phi}$ to denote parameter blocks and their horizon-wise collection.

\subsection{Simple Policy Gradient Approach}

A natural approach to implementing a policy gradient method in our context is to extend the approach by \citet{hamblyPolicyGradientMethods2021,hamblyPolicyGradientMethods2023}. Their work on policy gradient methods for linear quadratic regulators and linear-quadratic games captures the essential features of LQG control problems and games, namely the quadratic cost in the state and controls, and the linear dynamics of the state variable. It also provides strong convergence results. However, this work ignores some features that are important for financial applications, such as the interaction between the state and control in the cost function, first-order terms in the cost function, and the possibility of a mean reversion to a long-term mean different from 0 for the state. Hence, further work is necessary.

\subsubsection{Application to the LQG Game}

In the general case $\theta >0$, we derived an LQG game under an auxiliary measure and showed that its optimal controls are affine in the state, and that the value function is quadratic. We can therefore use the following parametrization\footnote{Note that this is a full parametrization. We could use the fact that $\eta^* = 0$ to substantially simplify the mathematical exposition and reduce the amount of numerical work.}:
\begin{align}
    h_k = \mathtt{D}_k X_k + \mathtt{d}_k
                        \qquad
    \bar\gamma_k = \mathtt{E}_k X_k + \mathtt{e}_k
                        \qquad
    \bar\eta_k = \mathtt{F}_k X_k + \mathtt{f}_k
\end{align} 
for $k=0,\ldots,K-1$, where $\mathtt{D}_k \in \mathbb{R}^{m \times n}$, $\mathtt{d}_k \in \mathbb{R}^{m}$, $\mathtt{E}_k \in \mathbb{R}^{d \times n}$, $\mathtt{e}_k \in \mathbb{R}^{d}$, $\mathtt{F}_k \in \mathbb{R}^{m \times n}$, $\mathtt{f}_k \in \mathbb{R}^{m}$. Let $\phi_k := \left(\mathtt{D}_k,\mathtt{d}_k,\mathtt{E}_k,\mathtt{e}_k,\mathtt{F}_k,\mathtt{f}_k\right)$. Further define $\mathtt{\mathbf{D}} := \left(\mathtt{D}_0, \ldots, \mathtt{D}_{K-1} \right)$ with similar definitions for $\mathtt{\mathbf{d}}, \mathtt{\mathbf{E}}, \mathtt{\mathbf{e}}, \mathtt{\mathbf{F}}, \mathtt{\mathbf{f}}$, and let $\boldsymbol{\Phi} = \left(\mathtt{\mathbf{D}},\mathtt{\mathbf{d}},\mathtt{\mathbf{E}},\mathtt{\mathbf{e}},\mathtt{\mathbf{F}},\mathtt{\mathbf{f}}\right)$.

Based on \eqref{eq:criterion:I:Pbar:step2} and \eqref{eq:criterion:I:Pbar}, the objective function $C$, as a function of the policy parameters, is
\begin{align}
    & C(\mathtt{D},\mathtt{d},\mathtt{E},\mathtt{e}, \mathtt{F},\mathtt{f}) 
                                \\
    =& \mathbf{E}^{\bar\gamma, \bar\eta} 
        \Bigg[  \theta \sum_{k=0}^{K-1} 
        g(\mathtt{D},\mathtt{d},\mathtt{E},\mathtt{e}, \mathtt{F},\mathtt{f};X_k)\Delta t 
        \Bigg]
                                \\
    =& \mathbf{E}^{\bar\gamma, \bar\eta} \Bigg[ 
        - \theta \left(R_T(\mathtt{D},\mathtt{d},\mathtt{E},\mathtt{e}; X_{0:K-1}) - R_0\right) 
        -\frac{1}{2} \sum_{k=0}^{K-1} \left\{\left(\mathtt{E}_k X_k + \mathtt{e}_k\right)' \left(\mathtt{E}_k X_k + \mathtt{e}_k\right) \Delta t 
        \right.
                            \nonumber\\
    & \left.
        + \left(\mathtt{F}_k X_k + \mathtt{f}_k\right)' \Xi_k^{-1} \left(\mathtt{F}_k X_k + \mathtt{f}_k\right)\right\}
    \Bigg],
\end{align}
which we also denote more succinctly as $C(\boldsymbol{\Phi})$ when we do not need to identify each parameter. With this notation, Hambly et al.'s so-called \emph{Natural Policy Gradient} (extended to include a zeroth-order term) for a sequence of episodes $\ell=1,\ldots, M$ produces the following \emph{online} update: 
\begin{align}
    \phi^{(\ell+1)}_k = \phi^{(\ell)}_k \pm \delta_\ell \, \widehat{\nabla}_{\phi_k} C\!\left(\boldsymbol{\Phi}^{(\ell)}\right) \, \left(\widehat{\Lambda}^{(\ell)}_k\right)^{-1},           
\end{align}
for $k=0,1,\ldots,K-1$ and where $\delta_\ell$ is the step size. In the case of unknown model parameters, $\widehat{\nabla}_{\phi_k} C\!$ denotes the empirical policy gradient of $C$ with respect to the parameter block $\phi_k$, and $\widehat{\Lambda}^{(\ell)}_k$ is an empirical estimate for the state covariance matrix $\Lambda_k := \mathbf{E}\left[ X_k X_k'\right]$ at time $k$. Here, the preconditioning by $(\widehat{\Lambda}^{(\ell)}_k)^{-1}$ is understood blockwise on the coefficients multiplying $X_k$; for the zeroth-order terms, the corresponding natural-gradient scaling is defined analogously (equivalently, by augmenting the state with a constant feature). 

The update sign $\pm$ is understood componentwise by parameter block. It is negative (descent) for the minimizing controller parameters $(\mathtt D,\mathtt d)$ and positive (ascent) for the adversarial parameters $(\mathtt E,\mathtt e,\mathtt F,\mathtt f)$. Thus, the learning problem is a finite-dimensional saddle-point problem in the policy parameters, with minimization over $(\mathtt D,\mathtt d)$ and maximization over $(\mathtt E,\mathtt e,\mathtt F,\mathtt f)$.

\subsubsection{Application to the Kelly Portfolio}

When $\theta \to 0$, we obtain the Kelly portfolio, as discussed in Section \ref{sec:KellyandFKS:Kelly}. Here, the implementation is simpler than in the general case for three reasons. First, the problem is set under the initial measure $\mathbb{P}$. Second, the problem is a standard LQG control problem, not a game. Third, the optimal control does not depend on the parameters of the value function.

As was the case for the LQG game, the optimal control is affine in the state. We can therefore use the following affine parametrization to set up the policy gradient method:
\begin{align}
    h^\text{Kelly}_k = \mathtt{D}^{\text{Kelly}}_k X_k + \mathtt{d}^{\text{Kelly}}_k
\end{align} 
for $k=0,\ldots,K-1$, where $\mathtt{D}^{\text{Kelly}}_k \in \mathbb{R}^{m \times n}$, $\mathtt{d}^{\text{Kelly}}_k \in \mathbb{R}^{m}$, and $\phi^{\text{Kelly}}_k := \left(\mathtt{D}^{\text{Kelly}}_k, \mathtt{d}^{\text{Kelly}}_k\right)$. As earlier, we define $\mathtt{\mathbf{D}}^{\text{Kelly}} := \left(\mathtt{D}^{\text{Kelly}}_0, \ldots, \mathtt{D}^{\text{Kelly}}_{K-1} \right)$ with similar definitions for $\mathtt{\mathbf{d}}$, and we let $\boldsymbol{\Phi}^{\text{Kelly}} = \left(\mathtt{\mathbf{D}}^{\text{Kelly}},\mathtt{\mathbf{d}}^{\text{Kelly}}\right)$.

Based on \eqref{eq:valuefunction:u:Kelly} and \eqref{eq:J:Kelly}, the objective function $C$, as a function of the policy parameters, is
\begin{align}
    & C^{\text{Kelly}}(\mathtt{D}^{\text{Kelly}},\mathtt{d}^{\text{Kelly}}) 
                                \\
    =& \mathbf{E} 
        \Bigg[ \sum_{k=0}^{K-1} 
        g^{\text{Kelly}}(\mathtt{D}^{\text{Kelly}},\mathtt{d}^{\text{Kelly}};X_k)\Delta t 
        \Bigg]
                                \\
    =& \mathbf{E} \left[
            \left(R_T(\mathtt{D}^{\text{Kelly}},\mathtt{d}^{\text{Kelly}}; X_{0:K-1}) - R_0\right)  
        \right]
\end{align}
Then Hambly et al.'s so-called \emph{Natural Policy Gradient} updating rule (extended here as well to include a zeroth-order term) for a sequence of episodes $\ell=1,\ldots,M$ is: 
\begin{align}
    \phi_k^{\text{Kelly},(\ell+1)} = \phi_k^{\text{Kelly},(\ell)} + \delta_\ell \,\widehat{\nabla}_{\phi_k} C^{\text{Kelly}}\!\left(\boldsymbol{\Phi}^{\text{Kelly},(\ell)}\right)\, \left(\widehat{\Lambda}^{(\ell)}_k\right)^{-1},           
\end{align}
for $k=0,1,\ldots,K-1$. Here again, the preconditioning by $(\widehat{\Lambda}^{(\ell)}_k)^{-1}$ is understood blockwise on the coefficients multiplying $X_k$; for the zeroth-order term, the corresponding natural-gradient scaling is defined analogously (equivalently, by augmenting the state with a constant feature). This is a gradient-ascent update for the maximizing controller parameters $(\mathtt D^{\text{Kelly}},\mathtt d^{\text{Kelly}})$.

\subsection{From Policy Gradient to Actor-Critic}

Actor-critic methods build on this approach to update both the policy and the value function. The \emph{actor} performs the three policy gradient steps to update the parameters of the policy, at the current estimate for the value function, while the \emph{critic} performs three similar steps to update the parameters of the value function, at the current estimate for the policy.

In the general case of the LQG game, we take the critic as
\begin{align}
u_k(X_k)= \frac{1}{2} X_k' \mathtt{P}_k X_k + X_k'\mathtt{p}_k + \mathtt{r}_k,
\end{align}
where $\mathtt{P}_k \in \mathbb{R}^{n \times n}, \mathtt{p}_k \in \mathbb{R}^{n}$, and $\mathtt{r}_k \in \mathbb{R}$.

The critic parameters $(\mathtt P_k,\mathtt p_k,\mathtt r_k)$ are updated on a faster timescale by \emph{temporal-difference} methods. The key here is to establish that the joint gradient updates will converge to the optimal controls and the optimal value function.

Finally, we note that global convergence results exist for policy-gradient methods for LQR problems--- see \citep{fazelGlobalConvergencePolicy2018,hamblyPolicyGradientMethods2021} for control problems and \citet{hamblyPolicyGradientMethods2023} for linear-quadratic games. On the actor-critic front, \citet{alacaogluNaturalActorCriticFramework2022} obtained promising results for two-player Markov games. 

In the special case of the Kelly portfolio, the policy does not depend on the value function. Hence, adding a critic does not improve the speed or accuracy of the policy estimation. Therefore, a critic only needs to be included if an estimate of the value function is of independent interest to the investor.

\section{Conclusions}

We have considered a risk-sensitive benchmarked asset management problem in the form of a non-standard LEQG problem. The problem falls outside the classic LEQG template because the state process is uncontrolled, and the noise appears linearly within an exponential. The assets and the benchmark are modeled as continuous-time processes, but portfolio rebalancing occurs at a fixed discrete time schedule, leading to a discrete-time stochastic control problem. Furthermore, the dynamics of the assets and the benchmark are affected by a factor process that can be observed at the discrete rebalancing times. In this discrete-time setup, it was rather straightforward to introduce randomized controls for exploration.

We applied the free energy-entropy duality to connect the risk-sensitive control problem, set with respect to an initial measure, to a risk-neutral randomized stochastic game, defined with respect to a transformed measure. This approach automatically implies a penalization in the risk-neutral problem. The penalization is given by the relative entropy of the transformed measure with respect to the original one. This connection also offers a theoretical justification for penalizing the randomization of controls, which has appeared more ad hoc in the literature.

We have solved the resulting risk-neutral penalized stochastic game, assuming the model coefficients are known. We show that, as functions of the observable factors, the optimal controls are linear and the optimal value function is quadratic. We also considered the limiting case in which the risk-sensitivity parameter tends to zero. This leads to a Kelly portfolio. We have shown that the optimal asset allocation can be rewritten as a rescaled fractional Kelly strategy. We also discussed various implications of the assumptions underlying our results.

For the case where the model parameters are not known, we discussed how to adapt our approach using known Reinforcement Learning methods. In particular, we focused on policy gradient and actor-critic methods. The structure we obtained for the optimal strategies and value, with their linear and quadratic dependence on the factors, indicates the parametrization to use for the controls and the value function. This supports applying policy gradient and actor-critic methods to learn the parameters from market observations.

One fundamental study in risk-sensitive stochastic control appears in the paper by \citet{kuna02}. The authors use a change of measure to transform a non-standard discrete-time risk-sensitive LEQG problem, like ours, into a standard discrete-time LEQG problem. This transformation may lead to an alternative two-step approach to our problem. First, the Kuroda-Nagai methodology transforms the non-standard LEQG problem into a standard one. Second, applying the free energy-entropy duality reduces the problem to an equivalent LQG problem. Such a two-step approach has been studied in detail by \citet{LleoRunggaldier_EntropyRegularizationinRLandRSIM_2026}, which is formulated entirely in continuous time and does not consider randomized controls from the outset. It is shown there that a one-step approach, as used here, yields the same results as the two-step approach. We could apply a two-step approach to the present setup, following another of our papers, \citet{LleoRunggaldier_ExploratoryRandomization_2026}, but omit it to avoid overburdening this paper.

%
%


\end{document}